\documentclass[a4paper,twocolumn,11pt,accepted=2024-07-08,allowtoday]{quantumarticle}

\pdfoutput=1
\usepackage{graphicx}
\usepackage{dcolumn}
\usepackage{subcaption}
\usepackage{bm}
\usepackage{amssymb}
\usepackage{amsmath}
\usepackage{amsfonts}
\usepackage{epsfig}
\usepackage{mathbbol}
\usepackage{latexsym}
\usepackage{theorem}
\newtheorem{theorem}{Theorem}
\newenvironment{proof}{\noindent \textbf{{Proof~} }}{\hfill $\blacksquare$}
\usepackage{dcolumn}
\usepackage{hyperref}
\usepackage{float}
\usepackage[normalem]{ulem}
\usepackage{xcolor}
\usepackage{epstopdf}
\usepackage{physics}
\usepackage{lipsum}
\usepackage{microtype}
\usepackage[sort&compress,numbers]{natbib}
\usepackage[utf8]{inputenc}

\begin{document}
	
	\title{Learning quantum phases via single-qubit disentanglement}
	
	\author{Zheng An}
	
	\affiliation{Institute of Physics, Beijing National Laboratory for Condensed
		Matter Physics, Chinese Academy of Sciences, Beijing 100190, China}
	
	\affiliation{School of Physical Sciences, University of Chinese Academy of
		Sciences, Beijing 100049, China}
	
	\affiliation{Department of Physics, The Hong Kong University of Science and
		Technology, Clear Water Bay, Kowloon, Hong Kong, China}
  
	\orcid{0000-0003-4954-2671}
 
	\author{Chenfeng Cao}
	
	\affiliation{Department of Physics, The Hong Kong University of Science and
		Technology, Clear Water Bay, Kowloon, Hong Kong, China}

        \orcid{0000-0001-5589-7503}
        
	\author{Cheng-Qian Xu}
	
	\affiliation{Institute of Physics, Beijing National Laboratory for Condensed
		Matter Physics, Chinese Academy of Sciences, Beijing 100190, China}
	
	\affiliation{School of Physical Sciences, University of Chinese Academy of
		Sciences, Beijing 100049, China}
  
	\orcid{0000-0003-2821-1216}
 
	\author{D. L. Zhou} \email[]{zhoudl72@iphy.ac.cn}
	
	\affiliation{Institute of Physics, Beijing National Laboratory for Condensed
		Matter Physics, Chinese Academy of Sciences, Beijing 100190, China}
	
	\affiliation{School of Physical Sciences, University of Chinese Academy of
		Sciences, Beijing 100049, China} 
	
	\date{\today}

        \maketitle
	
	\begin{abstract}
		Identifying phases of matter presents considerable challenges, particularly within the domain of quantum theory, where the complexity of ground states appears to increase exponentially with system size. Quantum many-body systems exhibit an array of complex entanglement structures spanning distinct phases. Although extensive research has explored the relationship between quantum phase transitions and quantum entanglement, establishing a direct, pragmatic connection between them remains a critical challenge. In this work, we present a novel and efficient quantum phase transition classifier, utilizing disentanglement with reinforcement learning-optimized variational quantum circuits. We demonstrate the effectiveness of this method on quantum phase transitions in the transverse field Ising model (TFIM) and the XXZ model. Moreover, we observe the algorithm's ability to learn the Kramers-Wannier duality pertaining to entanglement structures in the TFIM. Our approach not only identifies phase transitions based on the performance of the disentangling circuits but also exhibits impressive scalability, facilitating its application in larger and more complex quantum systems. This study sheds light on the characterization of quantum phases through the entanglement structures inherent in quantum many-body systems.
	\end{abstract}
	
	\section{Introduction}
	A phase diagram illustrates qualitative variations in many-body systems as a function of physical system parameters~\cite{sachdev_2011}. The purpose of a phase classifier is to ascertain the phase diagram by examining measurement results of observables. In the Landau theory of phase transitions, for instance, only the average measurement value of a single local variable, known as the local order parameter, is necessary to determine the phase diagram. Indeed, a discontinuity in a local order parameter or one of its derivatives signifies a phase transition.
	
	In addition to order parameters, entanglement, an essential quantum mechanical nonlocal property, has been proposed to  characterize quantum phase transitions and to identify topological orders~\cite{PhysRevB.82.155138,PhysRevA.66.032110,PhysRevLett.90.227902,Osterloh2002,PhysRevLett.96.110404,PhysRevLett.96.110405}. Entanglement entropy is  powerful for detecting critical points and studying universal properties in quantum phase transitions, especially in one-dimensional systems. At critical points, entanglement entropy exhibits logarithmic scaling with subsystem size, revealing crucial information about the underlying critical system~\cite{PhysRevLett.90.227902,Osterloh2002}. Topological entanglement entropy characterizes topologically ordered states, capturing non-local entanglement and serving as an effective diagnostic tool for detecting and characterizing these exotic phases~\cite{PhysRevB.82.155138,PhysRevLett.96.110404}. Quantum entanglement is believed to be crucial for understanding quantum phase transitions~\cite{Osterloh2002,RevModPhys.80.517,PhysRevLett.93.250404}, with distinct entanglement structures characterizing different quantum phases.

 \begin{figure}[tbhp]
		\centering
		\includegraphics[width=8.5cm]{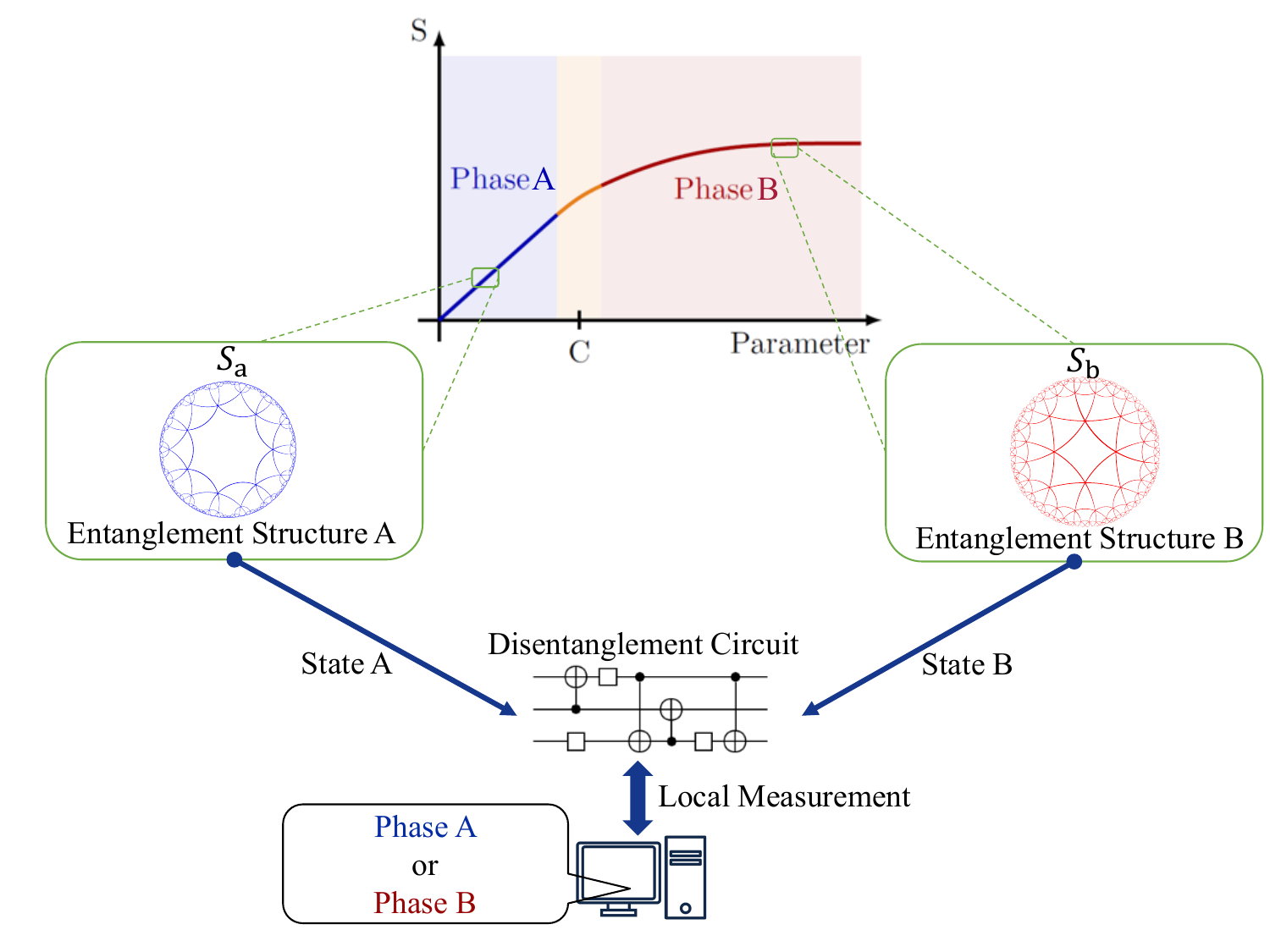}
		\caption{\label{fig:prob} Classifying quantum phases through single-qubit disentanglement. The local entanglement of a quantum many-body system's ground state may encompass intricate structures, with some of the primary structures potentially being closely related to the quantum phase. The RL agent's objective is to reduce the local entanglement of the ground state within a given quantum phase by employing limited operations and observations. Consequently, we can deduce the predominant entanglement structure present in the local entanglement, facilitating the classification of the quantum phase.}
	\end{figure}
 
	In our examination of quantum entanglement, we discovered a strong correlation between disentanglement and entanglement structure. Disentanglement is a process that transforms an initially entangled state of a composite system into a separable state. However, a universal disentangling machine is non-existent~\cite{PhysRevA.59.3320,PhysRevA.60.4341,PhysRevLett.83.1451}. Despite this, by relaxing the constraint that the reduced density matrices of subsystems remain perfectly unaffected, local operations can achieve approximate disentanglement~\cite{PhysRevA.61.052301,PhysRevA.61.032108,Yu2004}. Disentanglement's central concept is to decouple quantum correlations among subsystems. Consequently, we aim to investigate whether a quantum circuit-implemented unitary transformation can eradicate local entanglement within a quantum state. Moreover, we seek to explore the relationship between entanglement structure and circuit structure, particularly for states within the same phase.
	
	Reinforcement learning (RL) is a branch of machine learning that focuses on training agents to make optimal decisions in complex environments by interacting with these environments and receiving feedback in the form of rewards or penalties~\cite{sutton}. Over the years, RL has been successfully applied to a wide range of applications, from video games~\cite{atari,vinyals2019grandmaster} and board games~\cite{alphago,alphagozero,alphazero} to neural architecture search~\cite{DBLP:conf/iclr/ZophL17}. In recent years, RL has emerged as a powerful tool in quantum information science, where it has been applied to various quantum-related problems including quantum control~\cite{Bukov2017,gate,ladder}, quantum error correction~\cite{error}, and quantum circuit optimization~\cite{fosel2021quantum,PS2021,cao2022quantum}. 
	
	Variational quantum circuits have become instrumental in approximating both known and unknown unitary evolutions, as well as for studying ground state properties, as demonstrated by hybrid quantum-classical algorithms on noisy intermediate-scale quantum devices~\cite{variational2014,preskill2018quantum, cerezo2020variational2, bharti2021noisy, Cao_2023, cao2024exploiting}. These circuits have also been used to characterize quantum phase transitions, applying established order parameters~\cite{PhysRevA.104.062614, PhysRevX.12.041035, lively2024, Lukas_Bosse_2024} or identifying unknown ones through machine learning~\cite{Cao:2024kea}. In this work, we explore from the disentanglement perspective and introduce an reinforcement learning approach to classify quantum phase transitions. Our method employs RL algorithm~\cite{atari} to design and optimize variational quantum circuits, which efficiently disentangle the ground state of quantum systems (Fig.~\ref{fig:prob}). By assessing the performance of these disentangling circuits over a range of system parameters, we can identify the critical point that demarcates distinct quantum phases. For practical purposes, the circuit strives to disentangle a single qubit from the rest of the subsystems within a constrained operation space.
	
	We demonstrate the effectiveness of our approach on two paradigmatic models of quantum magnetism: the transverse field Ising model (TFIM) and the XXZ model. For both models, our method accurately identifies the critical point and phase transition. Moreover, we find similarities in the disentangling circuit structures for states within the same phase, shedding light on the relation between entanglement structure and circuit design. Our approach not only classifies quantum phases but also provides insight into their inherent entanglement properties. We observe the algorithm's ability to learn the Kramers-Wannier duality related to entanglement structures in the TFIM. Compared to existing methods like the variational quantum classifier~\cite{PhysRevA.102.012415}, our approach requires significantly fewer resources, using only local measurements and a few layers of variational quantum circuits. We also show that disentangling circuits trained on small system sizes can be applied to much larger systems, exhibiting the scalability of our method.
	
	This paper is organized as follows: Section~\ref{sec:method} provides a succinct introduction to the identification of quantum phase transitions through disentanglement and delves into the reinforcement learning framework underpinning our approach. Section~\ref{sec:result} showcases the numerical results of our method, highlighting its ability to classify phases of the TFIM and XXZ models while illustrating the properties of the derived disentangling circuits. In Section~\ref{sec:comp}, we draw comparisons between our method and the variational quantum classifier. Section~\ref{sec:scal} examines the scalability of our approach. Section~\ref{dis} discusses the scope of our work, its limitations, and potential extensions to its application. Finally, Section~\ref{Conclusion} offers concluding remarks and explores prospective avenues for future research in this domain.
	
	\section{Methdology}\label{sec:method}
	
	\subsection{Phase transitions via disentanglement} 
	In the study of quantum phase transitions, disentanglement plays a crucial role in understanding the intricate nature of entangled states and their transformations across critical points. These transitions, characterized by non-analytic changes in ground state properties, are often accompanied by alterations in entanglement structures. Designing disentangling circuits poses a significant challenge, as it involves creating a specific unitary transformation as a finite sequence of unitary operators $A_t$ chosen from a universal set of gates $\mathcal{G}$. In this study, we task the agent with developing a disentanglement unitary operation $U$ that can effectively eliminate a target local entanglement entropy. Consequently, the agent's objective is to identify a unitary gate $U$, derived from the composition of gates in the sequence, which eradicates the target entanglement entropy.
	
	We employ entanglement entropy $S$ to quantify the entanglement present in quantum systems' ground states. When a bipartite quantum system $AB$ exists in a pure state, entanglement entropy can be calculated using the reduced density matrices $\rho_{A}$ or $\rho_{B}$, as shown in the following formula:
	
	\begin{equation}
		S \equiv -\operatorname{tr}\left(\rho_{A} \log_ {2} \rho_{A}\right) = -\operatorname{tr}\left(\rho_{B} \log_ {2} \rho_{B}\right).
	\end{equation}
	
	For subsystems $A$ or $B$ represented as spin-$\frac{1}{2}$ systems, $S$ ranges from $S=0$ (product state) to $S=1$ (maximally entangled state). For the ground state of the two investigated models, we consider one site as subsystem $A$ and the remainder as subsystem $B$.
	
	In this work, we introduce a training scheme for identifying distinct phases and their critical point. For a family of ground states $\rho_{\alpha}$, the primary concept involves training a reinforcement learning (RL) agent on two states $\rho_{\alpha=a}$ and $\rho_{\alpha=b}$, while presuming the existence of an undetermined critical point $a < c < b$ that categorizes the states into two groups. The aim is to pinpoint the accurate critical point $c$ by training the RL agent to achieve disentangling circuits until convergence with circuit layer $p$.
	
	We proceed to evaluate the performance of these disentangling circuits across the parameter range $(a, b)$, considering the total performance ($S_a$ and $S_b$) in relation to the proposed critical point $c$ as $S_a(c)$ and $S_b(c)$. We will demonstrate that the performance of the two RL-ansatz disentangling circuits with deep layers is inversely related, with the cross point at the precise critical point $c$, i.e., $S_a(c) = S_b(c)$.
	
	The circuit performance can be comprehended as follows: we assume that entanglement possesses two distinct structures in regimes below $c$ (phase $\mathcal{A}$) and above $c$ (phase $\mathcal{B}$), and that the circuits can differentiate and eliminate them. We refer to these diverse structures as features. When setting the states in the parameter range $(a, c')$ where $c' < c$, the circuit identifies the feature of phase $\mathcal{A}$ and accurately removes those entanglements. Nonetheless, when in the range $(c', c)$ of phase $\mathcal{A}$, the circuit encounters states with different features, but the majority feature persists from $a$ to $c'$. The $S_a$ will increase, signifying that some entanglements cannot be eradicated. When in the range $(c, b)$, the majority feature shifts, and $S_a$ swiftly elevates until saturation. The behavior of $S_b$ is precisely the reverse of $S_a$.
	
	If phases dictate entanglement structures, the feature count for both phases will be equal when situated at the critical point $c$, i.e., $S_a(c) = S_b(c)$. Conversely, if $S_a(c) \neq S_b(c)$, then the intersection cannot represent the critical point $c$, as the phase spanning from the cross point to the critical point would constitute another phase.
	
	\subsection{Reinforcement learning and variational quantum disentangler}
	\begin{figure*}[!htbp]
		\centering
		\includegraphics[width=16.5cm]{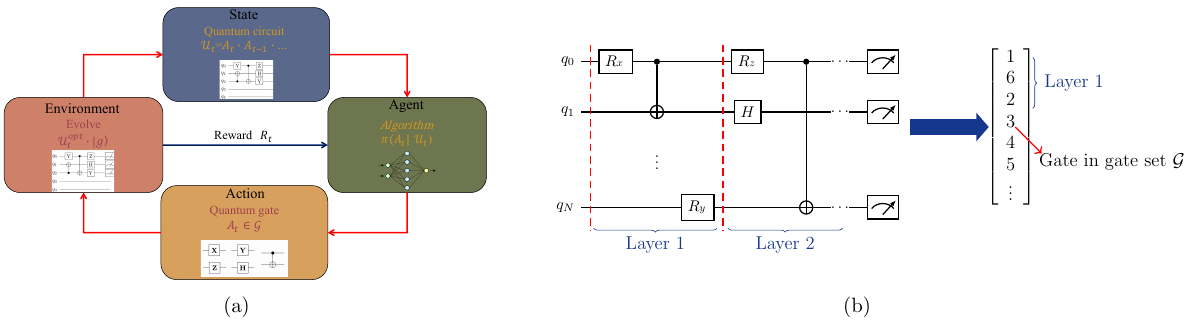}
		\caption{The deep reinforcement learning (DRL) architecture consists of two primary components: (a) The DRL environment, which is depicted as a quantum circuit modeled by the optimized variational quantum circuit $U^{opt}_t$ with an initial ground state $|g\rangle$ at each episode. During each time step $t$, the agent acquires the current state of $U_t$ and, based on this information, selects the subsequent gate $A_t$ from the gate set $\mathcal{G}$ to apply to the quantum circuit. Consequently, the environment returns the real-valued reward $R_t$ to the agent, which is a function of the entanglement entropy $S$. The agent's policy $\pi$ is encoded within a deep neural network. (b) State representation for the RL agent: At each time step, the agent receives an input vector representing the architecture of the quantum circuit $U_t$. This circuit is depicted as an array, with each element symbolizing a quantum gate in the given gate set. A layer of a quantum circuit is defined as a collection of exactly one quantum gate acting on each qubit within the selected range. Two-qubit gates (CNOTs) are only counted for the target qubit and so each qubit in the selected range is either acted upon by a single-qubit gate or is the target qubit of a CNOT.}\label{fig:flow}
	\end{figure*}
	
	Reinforcement learning (RL) is a machine learning paradigm in which an intelligent agent learns to make decisions by interacting with an environment (see Fig.~\ref{fig:flow}). The agent performs actions within the environment and receives feedback in the form of rewards and state (Representation of quantum circuit), which are used to adjust the agent's behavior. In this process, the agent seeks to maximize the reward per step. The agent's knowledge about the environment and its decision-making strategy, referred to as its policy, are typically represented by a mathematical function or a deep neural network (deep reinforcement learning). The policy dictates the agent's choice of actions based on the current state of the environment. By optimizing its policy to maximize the accumulated rewards, the agent learns to navigate the complex dynamics of the environment and make strategic decisions. In the context of our paper, reinforcement learning enables agents to efficiently select and apply quantum gates to optimize entanglement properties, furthering our understanding of quantum systems.
	
	In our framework, the environment's state consists of a quantum circuit initialized as an identity at the beginning of each episode. The agent incrementally constructs the circuit at every step by selecting a gate from $\mathcal{G}$ according to the policy $\pi$ encoded in the deep neural network, as depicted in Fig.~\ref{fig:flow} (a).
	
	We consider a gate set $\mathcal{G}$ that includes single-qubit rotations (along the x, y, or z-directions), Hadamard, and controlled-NOT (CNOT) gates, which together form a universal gate set. Rotation gates are parameterized by continuous variables, while CNOT and Hadamard gates are fixed. All single-qubit rotation gate parameters are initialized to $\pi$. It is important to note that the operation space is limited, acting only on the nearest or next-nearest qubits of the target qubit.
	
	At each time step, the reinforcement learning agent is provided with an input vector that encapsulates the architecture of the quantum circuit $U_t$. This representation employs an array $s_t$, wherein each element corresponds to a quantum gate belonging to the predetermined gate set (see Fig.~\ref{fig:flow} (b)).  
	
	The agent's reward in each step is given by:
	\begin{equation}
		R_{t}=\begin{cases}
			0, & t\in\{0,1,\dots,N-1\}\\
			\mathcal{R}, & t=N
		\end{cases}
		\label{reward}
	\end{equation}
	Here, $\mathcal{R}$ represents the modified reward function for our learning task. The agent does not receive an immediate reward, but instead at time $N$. A quantum circuit grants the last step's reward with optimized parameters. After finalizing the quantum circuit's structure, a classical optimizer is employed to adjust the initial parameters of the parameterized rotation gate to those with the minimum entanglement entropy.
	
	We define $\mathcal{R}$ as:
    \begin{equation}
        \mathcal{R} = 
\begin{cases}
0, & \text{if } S_{RL} \geq S_0 \\
\dfrac{S_0 - S_{RL}}{S_0}, & \text{otherwise}
\end{cases}
\label{eq:reward2}
    \end{equation}
	where $S_{RL}$ denotes the final entanglement entropy achieved by our RL-assisted quantum algorithm, and $S_{0}$ represents the ground-state entanglement entropy provided by the physical system.
	
	We employ Deep Q-learning algorithms to train the agents in this work (details of the algorithm are in supplementary materials), relying on the reward function. For the classical optimizer, we utilize the BFGS algorithm~\cite{broyden1970convergence,fletcher1970new,goldfarb1970family,shanno1970conditioning} implemented in SciPy~\cite{Virtanen2020}, enabling us to determine the parameters for each gate, from $U_t$ to $U^{opt}_t$.
	
	\section{Numerical Results}\label{sec:result}
	
	\subsection{Transverse field Ising model} 
	We first apply our method to
	one-dimensional transverse-field Ising model. The Hamiltonian for the Ising
	model on a 1D lattice with N sites in a transverse field is given by
	\begin{equation}
		H=-\sum_{j=0}^{N-1}\left(\lambda \sigma_{j}^{x} \sigma_{j+1}^{x}+\sigma_{j}^{z}\right),
	\end{equation}
	where $\sigma_{j}^{a}$ is the $a$ component of Pauli matrix ($a=x,y$ or $z$)
	at site $j$, and $\lambda$ is the inverse strength of the external field. We
	assume the periodic boundary conditions, i.e., the $N$-th site is identified
	with the $0$-th site. As is well known, there is a quantum phase transition
	at the critical point $\lambda=1$ for the transverse Ising model, which means the
	entanglement properties of the ground state change dramatically when
	crossing this critical point.
	
	When $\lambda$ approaches zero, the TFIM ground state becomes a product of spins
	pointing in the positive $z$ direction,
	\begin{equation}
		|\psi_{\text{TFIM}}\rangle_{\lambda \rightarrow 0} = |0\rangle^{\otimes N},
	\end{equation}
	where $|0\rangle$ and $|1\rangle$ denote the single spin state with spin up and with
	spin down along $z$-direction respectively.
	
	In the $\lambda \rightarrow \infty$ limit, the ground state is an
	adiabatic continuation of the GHZ state~\cite{PhysRevB.91.125121}:
	\begin{equation}
		|\psi_{\text{TFIM}}\rangle_{\lambda \rightarrow \infty} = \frac{1}{\sqrt{2}} ( |+ \rangle^{\otimes N} +|-\rangle^{\otimes N}),
	\end{equation}
	where $|+\rangle$ and $|-\rangle$ denote the single spin state with spin up and with
	spin down along $x$-direction respectively.
	
	Under the above two limits, the entanglement entropy of the ground state
	behaves as
	\begin{align}
		\label{eq:2}
		S_{\lambda\rightarrow 0} & = 0, \\
		S_{\lambda\rightarrow \infty} & = 1.
	\end{align}

    \begin{figure}[tbh]
		\centering
		\includegraphics[width=0.85\columnwidth]{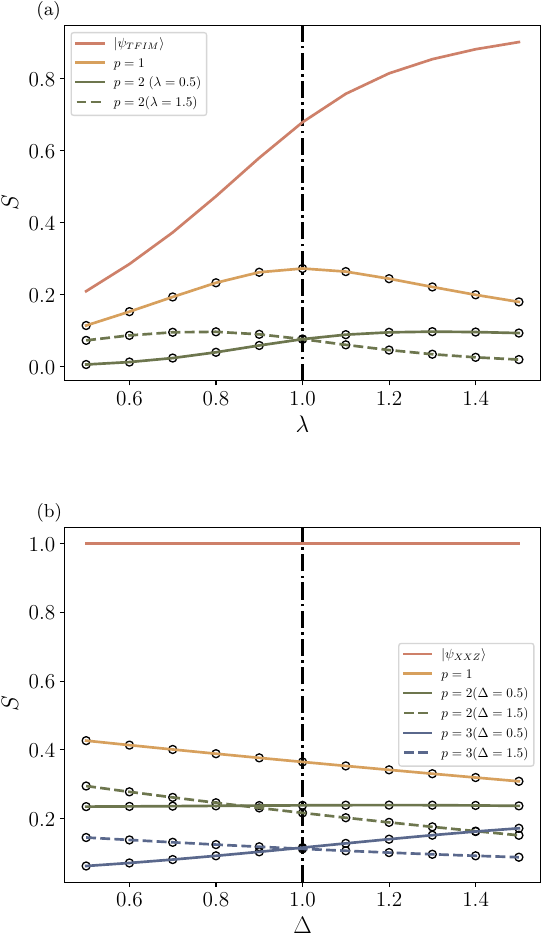}
		\caption{\label{fig:resluts} Resulting entanglement entropy of RL-designed disentangling circuits. $p$ represents the number of circuit layers. Solid and dashed lines indicate the disentangling circuits trained from different parameters. The models are for (a) TFIM and (b) XXZ model with 8 sites. The dash-dotted vertical line signifies the critical point.}
	\end{figure}
 
	At the critical point, $\lambda = 1$, there is a phase transition in the ground
	state in the thermodynamic limit~\cite{sachdev_2011}. The
	correlation function $\langle\sigma_{i}^{\alpha} \sigma_{j}^{\beta}\rangle-\langle\sigma_{i}^{\alpha}\rangle\langle\sigma_{j}^{\beta}\rangle$ decays
	polynomially as a function of separation at this point, while for all other
	values of $\lambda$, this decay is exponential.

	In Fig.~\ref{fig:resluts} (a), we show the results of the optimal variational
	algorithm to disentangle the ground state with different values of $\lambda$. We
	take two ground-states from two phases of the system to get the RL-ansatz
	quantum circuit. The operation space is set to the nearest qubits, i.e.,
	3 qubits. We also selected two training states asymmetrically 
	from two phases far from the critical point to exclude the effect 
	of selecting particular training points. (See supplementary 
	material for details on the robustness of the model.)
	
	The performance of the designed algorithm improves as the number circuit layers $p$ increases.
	As the RL-ansatz disentangling circuit with one layer ($p=1$),
	the transformation space is not large enough to cut the correlation with the
	two subsystems, so different training samples have the same results.
	However, unlike the raw local entanglement entropy case, the disentangle
	circuit entropy shows three regions with different values of $\lambda$. In both
	cases, the entanglement entropy rises quickly from a small number from point
	$\lambda=0.5$ to point $\lambda=0.9$ and then slowly rises to the maximum value around
	the critical point $\lambda=1$. The situation after the critical point $\lambda=1$ is
	precisely the opposite of before. So we can divide the space of $\lambda$ into
	``fast-slow-fast'' three regions according to the speed of entropy change.
	We find that the most significant entanglement is present in the parameter
	region close to the critical point. The slow region corresponds with the
	quantum critical region~\cite{sachdev_2011}. This indicates that the local quantum
	entanglement has a deep relation with the phase transition region.
	
	However, with the layer number larger than one, the results show different
	structures of circuits with different phases of $\lambda$. As Fig.~\ref{fig:resluts} (a)
	shows, the designed circuit structures with two layers have two optimal
	circuits in different phases. Also, in this case, the characteristic is
	obtained, and we detect the transition at $\lambda=1$. We confirm that the results
	with $p=2$ are indeed convergent (see Fig. S3 in Supplementary materials).
	
	\subsubsection{Understanding agent actions} 
	To confirm what agents have learned
	from disentanglement under TFIM, we first determine the optimal local
	disentanglement in theory. We proved that the minimum entropy of $\rho_P$
	decreased by unitary operation is determined by the initial state
	$\rho_{RP}$ of the region $RP$, on which the unitary operation acts. More
	specifically, we have the following theorem.
	
	\begin{theorem}
		Suppose the eigen-decomposition of the state $\rho_{RP}$ is
		$\rho_{RP} = \sum_{i=1}^{d_P d_R} p_i \ket{\psi_i} \bra{\psi_i}$, where $d_P$ ($d_R$) is
		the dimension of subsystem $P$ ($R$). Without loss of generality, we assume
		these eigenvalues are in decreasing order, i.e., $p_i \ge p_{i+1}$. By
		performing a unitary operation $U$ on the $RP$, the minimal
		entanglement entropy of the subsystem $P$ is
		\begin{equation}
			\min_U S(\rho_P(U)) = - \sum_m q_m {\rm \log_{2}} q_m,
		\end{equation}
		where $\rho_P(U) = {\rm Tr}_R \left[ U \rho_{RP} U^\dagger \right]$,
		$S(\rho) = - {\rm Tr} \left[ \rho {\rm \log_2} \rho \right]$ is the Von-Neumann
		entropy of state $\rho$, and $q_m = \sum_{j=1}^{d_R} p_{(m-1)d_R + j}$.
		\label{theorem1}
	\end{theorem}
	
	\begin{proof}
		The proof is given in the given in
		Supplementary materials, Section S-\uppercase\expandafter{\romannumeral1}.
	\end{proof}
	
	Although \textbf{Theorem}~\ref{theorem1} gives the optimal local disentanglement
	theoretically, it does not imply that we can distinguish quantum phases by
	directly applying the theorem. In fact, the relation between disentanglement
	and quantum phase is subtle. Because the optimal solution will destroy all
	entanglement structures equally, the minimal entanglement entropy in our
	problem is so small in both phases that it can not be used to distinguish
	the phases. However, the RL agent often finds a local minimum but not a
	global one, which leads it to destroy the dominant entangled structure under
	a specific phase. Therefore, even if RL cannot give an optimal strategy in
	many cases, it may fits the problem.
	
	Interestingly, we found the RL agent finds the optimal solution to
	disentangle the local entanglement of the nearest two-body in the phase of
	$\lambda >1$ (Fig.~\ref{fig:TFIMdua}), which is exactly the solution of
	\textbf{Theorem}~\ref{theorem1} when $RP$ is restricted to two nearest spins. At first
	glance, the RL-ansatz disentangling circuit in the phase of $\lambda <1$ may be a
	suboptimal disentanglement of a complex three-body local entanglement.
	However, we realize that there is a well known Kramers-Wannier duality in
	TFIM~\cite{PhysRev.60.252,PhysRevD.17.2637}. For the model, the self-dual
	point is the phase transition point $\lambda=1$. When we perform the duality
	transformation which defines new Pauli operators $\mu_{i}^{z}$ and $\mu_{i}^{x}$
	in a dual lattice
	\begin{align}
		\mu_{j}^{x}=\prod_{i \leq j} \sigma_{i}^{z} \\
		\mu_{j}^{z}=\sigma_{j+1}^{x} \sigma_{j}^{x},
	\end{align}
	and the original Hamiltonian can be written as
	\begin{equation}
		H=-\sum_{j=0}^{N-1}\left(\mu_{j}^{x} \mu_{j+1}^{x}+\lambda \mu_{j}^{z}\right).
	\end{equation}
	
	We found that the optimal disentanglement of the nearest neighbor under the duality transformation coincides with the results of the RL agent in phase $\lambda <1$
	(Fig.~\ref{fig:TFIMdua}).
	\begin{figure}[!htbp]
		\centering 
		\includegraphics[width=0.45\textwidth]{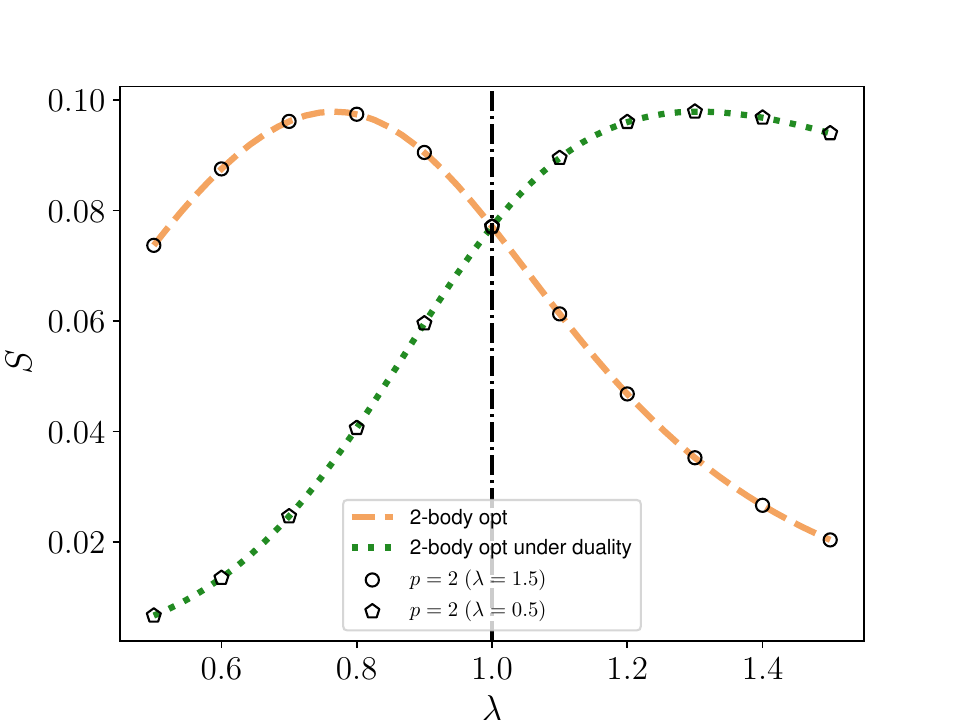}
		\caption{Entanglement entropy versus $\lambda$ for the transverse field Ising model (TFIM), comparing the reinforcement learning (RL)-ansatz disentangling circuit and the nearest two-body optimal disentanglement. 
			The lines depict the theoretically optimal outcomes, whereas the unfilled markers indicate the numerical results obtained from the quantum circuits designed based on two distinct ground states ($\lambda=0.5,1.5$). \label{fig:TFIMdua}}
	\end{figure}
	
	The difference between the
	result of the RL agent and the theoretical optimal solution is negligible.
	These results suggest that the agent learns the duality of the
	Hamiltonian, which is an important feature in classifying quantum
	phases. This further shows that the RL agent has found the main local
	entanglement structure between different phases.
	
	It is important to recognize that the inherent duality of the model is already known, serving as a fundamental component of our understanding of phase transitions. This established knowledge allows us to apply \textbf{Theorem}~\ref{theorem1} effectively to determine phase transition points. In contrast, our machine learning approach was not initially programmed with knowledge of this duality, yet it successfully identified the same critical point as established by theoretical analysis. This demonstrates the capability of our method to independently corroborate theoretical insights. In such cases, \textbf{Theorem}~\ref{theorem1} cannot leverage duality transformations to derive an optimal solution.
	
	It is worth noting that, for small areas, finding the optimal disentangler might be simple, allowing us to use an exact solution instead of a machine learning approximation. However, to do this, one must know the interaction area of the main entanglement structure. Discerning this structure within the context of a quantum system proves to be a nontrivial task, demanding an extensive array of calculations. In practice, the accurate determination of the phase transition point poses a considerable challenge for all currently established methods.

	\subsection{XXZ model} 
	The second model we apply our method to is the XXZ
	model. The Hamiltonian of the antiferromagnetic XXZ model on a 1D
	lattice with N sites is
	\begin{equation}
		H=\sum_{j=0}^{N-1}\left(\sigma_{j}^{x} \sigma_{j+1}^{x}+\sigma_{j}^{y}\sigma_{j+1}^{y}+\Delta\sigma_{j}^{z}\sigma_{j+1}^{z}\right),
	\end{equation}
	where $\Delta$ is a dimensionless parameter characterizing the anisotropy of
	the model. The XXZ spin chain has been widely
	studied~\cite{PhysRev.150.321,PhysRev.150.327,PhysRev.151.258,PhysRevA.2.1075,PhysRevA.3.786}.
	For the XXZ model, there exist two different phases at the ground state
	where $\Delta >0$, i.e., metallic phase, $0<\Delta\le1$, and
	antiferromagnetic phase, $\Delta >1$, which result from that the
	former is gap-less while the latter is gapped. The
	quantum phase transition is located at the isotropic point $\Delta=1$ which
	is the Kosterlitz-Thouless quantum phase transition (KT-QPT).
	
	First, we perform our learning scheme, using the same condition for the TFIM
	case. However, we found the results do not converge. So we set the
	operation space to include next-nearest qubits, i.e., 5 qubits.
	
	In Fig.~\ref{fig:resluts} (b), we show the results of the optimal circuits to
	disentangle the ground state with different values of $\Delta$. Similar to
	the TFIM model mentioned above, the performance of the designed circuits is
	proportional to the circuit depth $p$. As the RL-ansatz disentangling
	circuit with one layer ($p=1$), this case has a universal circuit structure
	for different values of $\Delta$, which is similar to the TFIM result.
	
	Nevertheless, unlike the TFIM case, the performance is different in the
	two-layer case. From Fig.~\ref{fig:resluts} (b), we observe interesting differences
	for disentangling circuit results at the crossing points. The intersection
	is far from the critical point because the number of layers of the quantum
	circuit is just between too shallow and enough layers. So we can see that
	there are two optimal solutions, but the intersection of the two is not the
	critical point. The above situation indicates that the ground-state
	entanglement structure of the XXZ model might be more complex than the TFIM
	case.
	
	When $p=3$, the optimal disentangling circuit is convergent (see
	Supplementary materials, Fig. S4), and the result is consistent with what we
	discussed in the former section. Also, we detect the transition at
	$\Delta=1$.
	
	\section{Comparison with variational quantum classifier}\label{sec:comp}
	In this section, we present an analytical comparison between our method and the variational quantum classifier, focusing on the perspective of resource requirements.
	
	The variational quantum classifier is a quantum algorithm designed for quantum phase classification. It is worth mentioning that numerous existing quantum classifiers have been proposed, with the majority of them being rooted in supervised learning~\cite{PhysRevA.101.032308,Cong2019}. This implies that additional information is incorporated during the training process. For the purpose of comparison, we have chosen variational quantum classifiers~\cite{PhysRevA.102.012415}. Much like our method, the variational quantum classifier utilizes a quantum circuit to identify quantum phase transitions. The training process of this technique involves a supervised learning approach, necessitating the advanced preparation of a substantial quantity of ground states as well as pairwise labels. Ultimately, the input verified quantum states are distinguished by voting on the observations of all qubits to determine the phase to which they belong.
	\begin{figure*}[!htbp]
		\centering
		\includegraphics[width=15cm]{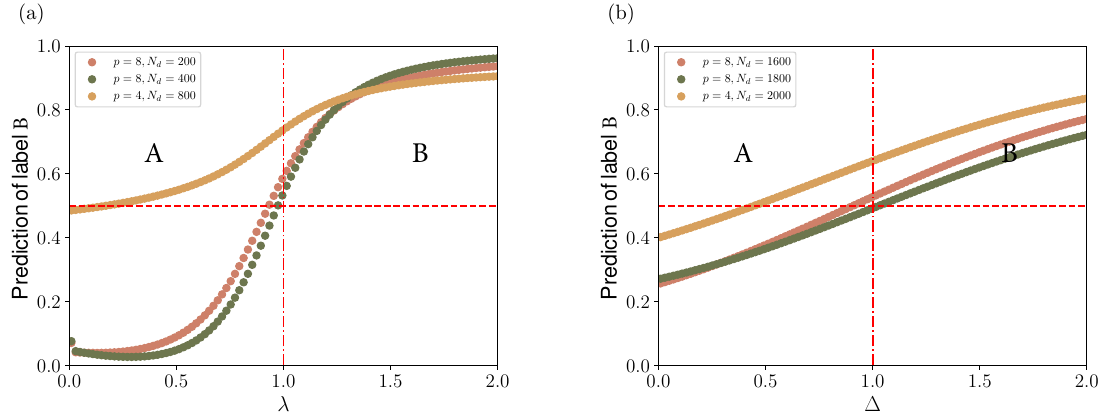}
		\caption{Performance of the variational quantum classifier: (a) TFIM with 8 sites; (b) XXZ model with 8 sites. $p$ represents the number of circuit layers, while $N_d$ denotes the quantity of training ground states. To ensure convergence of the results, we employed 1000 iterations of optimization for each data point. The predicted phase label corresponds to the expected probability that the classifier resides in a fixed quantum phase for a given state. When this probability reaches 0.5, we consider the system to be at a phase transition point.}
		\label{fig:vqc}
	\end{figure*}
	
	In this analysis, we simplify the variational quantum classifier method; the original version necessitates the use of a variational quantum eigensolver (VQE)~\cite{variational2014, PhysRevX.6.031007, PhysRevX.8.011021, PhysRevResearch.6.013205, liu2024variational} to simulate the ground state, while our focus lies solely on the classifier's performance. Consequently, we directly employ the system's ground state as input data. We implement this method to the two models previously examined in order to assess whether it demands fewer resources compared to our proposed approach.
	
	As illustrated in Fig.~\ref{fig:vqc}, the method demands a minimum of 8 layers in the quantum circuit and 200 and 1600 training ground states for the TFIM and XXZ models, respectively. The y-axis represents the probability of the algorithm predicting that a quantum state belongs to a specific phase in the absence of labels. A reliable algorithm can yield a probability of 0.5 near the critical point. In contrast, our method necessitates only two training ground states and a maximum of 3 layers in the quantum circuit. Furthermore, the method requires measurements of all qubits, whereas our approach mandates local measurements.
	
	\section{Scalability and robustness of the RL-ansatz disentangling circuit}\label{sec:scal}
	In previous discussions, the relationship between the entanglement structure, system size, and robustness was not explored. Further, our method only operates on
	the local space of the target qubit, which is easy to expand the scale of
	the system. So we want to study the effect of disentangling circuit training
	in the small system on a large system.
	
	In Fig.~\ref{fig:TFIMN} and Fig.~\ref{fig:XXZN}, we show the results of the
	disentangling circuit trained in 8-qubit systems on N-qubit problems with
	$N > 10$. For the RL-designed disentangling circuit, the circuit is obtained
	by a training process on a problem with qubit number $N = 8$ then applied to
	problems having a larger number of qubits ($N = 10,12,14$). The results show
	that our solution is not affected by the size of the system. This further
	implies that the entanglement structure is not affected by the size but is
	more closely related to the quantum phase.
	
	\begin{figure*}[!htbp]
		\centering
		\includegraphics[width=\textwidth]{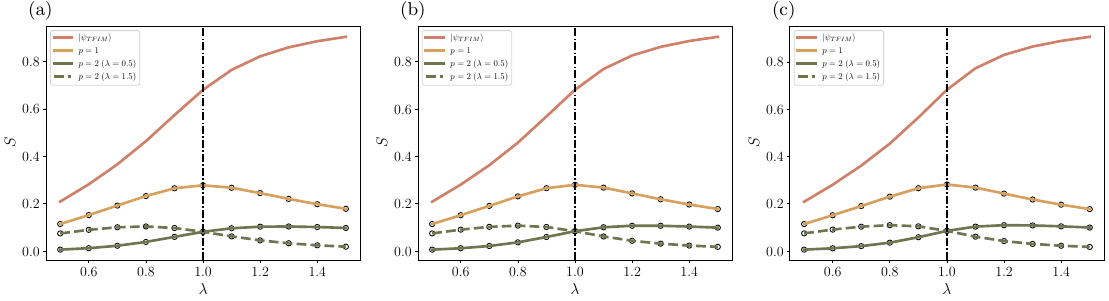}
		\caption{Performance of RL-designed disentangling circuits applied to TFIM
			with different numbers of qubits. In this plot, the RL-designed circuits
			are obtained by training on the problem with qubit number $N = 8$ then
			applied to problems with a larger number of qubits, (a):$N=10$ (b):$N=12$
			(c):$N=14$.}
		\label{fig:TFIMN}
	\end{figure*}
	
	\begin{figure*}[!htbp]
		\centering
		\includegraphics[width=\textwidth]{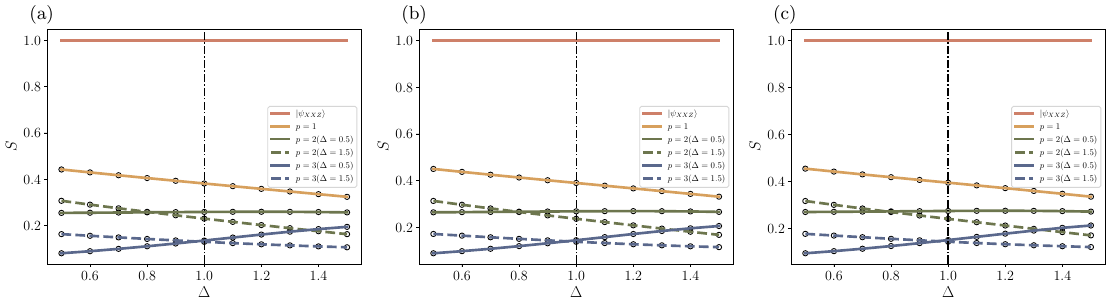}
		\caption{Performance of RL-designed disentangling circuits applied to XXZ model
			with different numbers of qubits. In this plot, the RL-designed circuits 
			are obtained by training on the problem with qubit number $N = 8$ then
			applied to problems with a larger number of qubits, (a):$N=10$ (b):$N=12$
			(c):$N=14$.}
		\label{fig:XXZN}
	\end{figure*}
	
	To verify the robustness of our RL-ansatz disentangling circuit with imperfect input states, we evaluated the performance of the training scheme using imperfect ground states generated by the VQE method. This method is often used in experimental settings where the exact ground state is not accessible. Through this evaluation, we aim to demonstrate the resilience and efficacy of our model, even when faced with suboptimal input data. Since obtaining the ground state of the TFIM is relatively straightforward and we have a theoretically optimal solution for this model, we chose to utilize the ground state of the XXZ model, which poses more significant challenges.
	
	\begin{figure*}[!htbp]
		\centering
		\includegraphics[width=15cm]{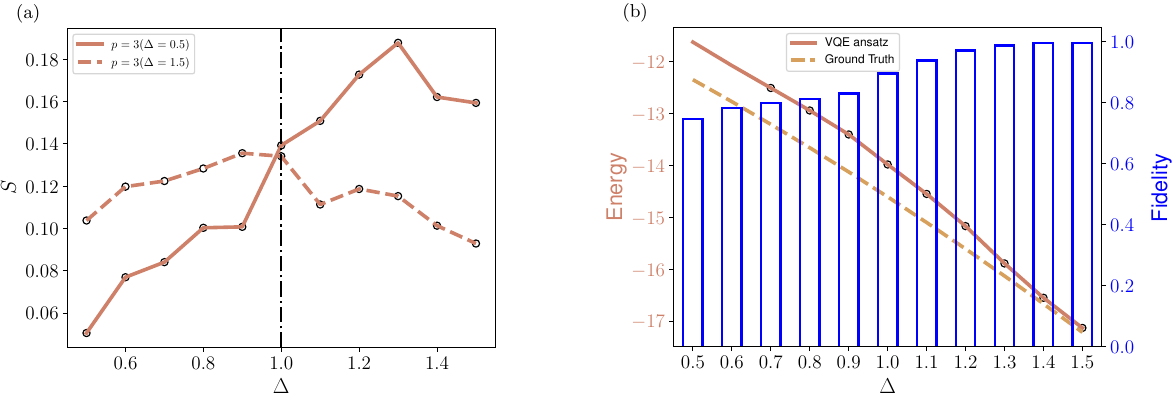}
		\caption{\label{fig:vqest}Performance of RL-designed disentangling circuits on Imperfect Ground States in the XXZ Model with $N=8$. (a) Resulting entanglement entropy of RL-designed disentangling circuits applied to the ground state prepared via the VQE. (b) Comparative analysis of energy and fidelity between the VQE-ansatz ground state and the actual ground state of the XXZ model. The energy of the ground state as predicted by the VQE ansatz (solid line) and the actual ground state (dashed line). The fidelity values between these states are represented as a blue outline bar chart.}
	\end{figure*}
	
	As shown in Fig.~\ref{fig:vqest}, our method effectively identifies phase transition points even at relative low fidelity ground states. This performance highlights the robustness of our methodology, confirming its ability to handle less-than-ideal conditions effectively. In other words, our approach can also be applied to detect phase transitions in scenarios where the ground state is imperfect.

	\section{Discussion}\label{dis}
	
	In this paper, we have presented an algorithm for studying entanglement structures and identifying quantum phases in many-body systems. The approach combines deep reinforcement learning-optimized variational quantum circuits, enabling a systematic investigation of entanglement patterns and offering valuable insights for both theoretical and experimental contexts.
	
	Our method serves distinct purposes compared to theoretically optimal disentanglement methods, such as those described in \textbf{Theorem}~\ref{theorem1}. While \textbf{Theorem}~\ref{theorem1} provides a mathematical framework to determine the minimal entanglement entropy achievable through unitary transformations, it does not specify the actual quantum state that corresponds to this minimal entropy. Consequently, \textbf{Theorem}~\ref{theorem1} lacks practical guidance on designing specific unitary operations for achieving this state. Our machine learning approach addresses this gap by not only predicting the minimal entropy but also facilitating the design of the corresponding disentangling unitary operations. This is crucial for practical applications where explicit construction of quantum circuits is necessary. Furthermore, our approach excels in identifying quantum phases and transitions by recognizing and disentangling key entanglement structures, thereby facilitating the detection of quantum phase transitions, theoretical disentanglement methods primarily focus on elucidating the theoretical underpinnings of these disentangling operations.
	Additionally, the application of RL methods enhances our understanding of entanglement structures and phase transitions across diverse systems. RL approaches are particularly resilient against system imperfections and data noise, capable of learning to identify critical features of phase transitions during the training process. Unlike traditional methods that rely on known order parameters like magnetization or correlation measures, RL can potentially discover novel indicators of phase transitions by directly analyzing data. This capability is especially advantageous in systems with poorly understood transitions or where traditional indicators are difficult to measure.
	
	From a theoretical standpoint, our algorithm sheds light on the understanding and characterization of quantum phases by leveraging the entanglement structures in quantum many-body systems. It not only allows for the identification of essential entanglement structures but also facilitates more efficient classification of quantum phases. Furthermore, the potential connections with measurement-induced phase transitions~\cite{PhysRevB.98.205136,PhysRevB.100.134306,PhysRevX.9.031009,PhysRevX.10.041020,PhysRevB.99.224307} present an exciting avenue for future research. By adapting our method to account for the effects of measurements on quantum many-body systems, we could uncover a more comprehensive understanding of the interplay between entanglement structures, quantum phase transitions, and measurement-induced phase transitions.
	
	For experimentalists, our algorithm provides a practical and scalable approach to phase identification in quantum many-body systems. It requires only local quantum operations and entanglement measures of local qubits, making it ideally suited for direct application in existing experiments, even without comprehensive knowledge of the quantum states. Utilizing a limited number of operations and measurements, the quantum circuit can be extended to accommodate large-sized systems. When attempting to measure the reduced density matrix in a many-body system, our algorithm can aid in identifying the quantum phase based on entanglement structures. Although performing such measurements can be challenging, indirect methods and experimental techniques are available for extracting information about the reduced density matrix or specific properties of interest, such as entanglement entropy or correlations.
	
	The choice of the region size for disentanglement depends on the entanglement structures' complexity and the specific many-body system under consideration. While our approach demonstrates success in the context of Ising spin chain systems, extending it to higher dimensions, critical systems, and topological systems may necessitate further development and customization. For instance, in topological systems, a more refined approach may be necessary to capture the intrinsic long-range entangled structures that cannot be disentangled via local unitary operations within a specific depth.
	
	Although finding optimal disentanglers analytically may be possible for small regions, the numerical approach using reinforcement learning offers several advantages. As the size of the region or the complexity of the entanglement structures increases, reinforcement learning provides scalability, adaptability, and robustness to imperfections, making it a versatile framework for quantum phase identification across various quantum many-body systems.
	
	\section{Conclusion}\label{Conclusion}
	
	In conclusion, we have proposed an innovative reinforcement learning approach for classifying quantum phase transitions through disentanglement. Our method employs deep Q-learning to design and optimize variational quantum circuits that effectively disentangle the ground state of a quantum system. By analyzing the performance of these disentangling circuits across a range of system parameters, we can identify the critical point that demarcates distinct quantum phases.
	
	We demonstrated the effectiveness of our approach on two paradigmatic models of quantum magnetism: the transverse field Ising model and the XXZ model. For both models, our method accurately identified the critical point and phase transition. Moreover, we found similarities in the disentangling circuit structures for states within the same phase, shedding light on the relation between entanglement structure and circuit design. Our approach not only classifies quantum phases but also provides insight into their inherent entanglement properties.
	
	Compared to existing methods, our approach requires significantly fewer resources, using local measurements and a few layers of variational quantum circuits. We also showed that disentangling circuits trained on small system sizes can be applied to much larger systems, exhibiting the scalability of our method. We expect that our approach may be extended to characterize
	more exotic quantum phases, e.g., a phase with topological order and a spin
	liquid. It provides a universal way to quantify entanglement structures of
	quantum phases. Our work highlights the promise of machine learning, and reinforcement learning in particular, for exploring and understanding the properties of complex quantum systems.

	\section{Acknowledgments} 
	This work is supported by the National
	Key Research and Development Program of China (Grant No. 2021YFA0718302 and
	No. 2021YFA1402104), NSF of China (Grants No. 11775300 and No. 12075310),
	and the Strategic Priority Research Program of Chinese Academy of Sciences
	(Grant No. XDB28000000).

\bibliography{main.bbl}
\bibliographystyle{quantum}

	~\\
\clearpage 
\onecolumngrid
\appendix
 \begin{center}
        \textbf{\large Supplemental Materials: Learning quantum phases via single-qubit disentanglement}
    \end{center}

    \setcounter{equation}{0}
    \setcounter{figure}{0}
    \setcounter{table}{0}
    \setcounter{page}{1}
    \renewcommand{\theequation}{S\arabic{equation}}
    \renewcommand{\thefigure}{S\arabic{figure}}
    \renewcommand{\thetable}{S\arabic{table}}
    \renewcommand{\thesection}{S-\Roman{section}}
    
     \section{Optimal local disentanglement}
    
    In this part, we aim to prove Theorem~1 in the main text. To this end,
    let's review the theory of majorization first, which plays an important role in
    quantum information theory.
    
    Suppose we have two $n$-dimensional vectors $p = (p_1, p_2, \cdots, p_n)$
    and $q = (q_1, q_2, \cdots, q_n)$. Without loss of generality, we assume
    that $p_{i} \ge p_{i+1}$ and $q_i \ge q_{i+1}$ for $i = 1,2,\cdots,n-1$.
    The theory of majorization gives a way to compare these two vectors:
    
    \textbf{Definition~S1}. Vector $p$ is majorized by $q$, denoted as $ p \prec q$, if
    \begin{equation}
        \sum_{i=1}^{k} p_i \le \sum_{i =1}^k q_i
    \end{equation}
    holds for every $k = 1,2,\cdots,n$.
    
    From the Definition~S1, we can see that if $ p \prec q$ then vector $p$ is more
    disordered in the distribution of components than vector $q$. Indeed, we have
    a theorem~\cite{Bahtia1997}:
    
    \textbf{Theorem~S1}. $p \prec q$ if and only if $p = D q$, where $D$ is doubly stochastic
    matrix, i.e., $D$ has non-negative entries and each row and column sums to $1$.
    
    Now we can prove Theorem~1 in the main text.
    Suppose state $\ket{\Phi}$ is a pure state of system $RPQ$. Then the reduced
    density matrices  of subsystem $RP$ is defined as
    \begin{equation}
        \rho_{RP} = {\rm Tr}_{Q} \left[ \ket{\Phi} \bra{\Phi} \right],
    \end{equation}
    When we perform a local unitary operation $U$ acting on subsystem $RP$, the reduced state of subsystem $P$ becomes
    \begin{equation}
        \label{eq:1}
        \rho_P(U) = {\rm Tr}_{R} \left[ U \rho_{RP}
        U^\dagger \right]
    \end{equation}
    By choosing a proper unitary transformation we aim to minimize the entropy
    \begin{equation}
        \label{eq:2}
        \min_{U} S(\rho_{P}(U))
    \end{equation}
    where $S(\rho) = - {\rm Tr} \left[ \rho {\rm \log_2} \rho \right]$, and $S(\rho_{P})$
    characterizes the entanglement between subsystem $P$ and subsystem $RQ$.
    
    Let the dimension of the Hilbert space of subsystem $P$ be $d_{P}$, and let the dimension of the Hilbert space of subsystem $R$ be $d_R$.
    We start with the eigen-decomposition of the state $\rho_{RP}$
    \begin{equation}
        \rho_{RP} = \sum_{i=1}^{d_P d_R} p_i \ket{\psi_i} \bra{\psi_i},
    \end{equation}
    where $p_i$ is the $i$-th eigenvalue, and $\ket{\psi_i}$ is the corresponding eigenstate.
    Without loss of generality, we assume these eigenvalues are in decreasing
    order, i.e., $p_i \ge p_{i+1}$. Then, for $1\le m\le d_P$ let
    \begin{equation}
        \label{eq:3}
        q_m = \sum_{j=1}^{d_R} p_{(m-1)d_R + j}.
    \end{equation}
    
    Now we are ready to give the following theorem on optimal local
    disentanglement.
    
    \textbf{Theorem~S2}: The minimal entanglement entropy
    \begin{equation}
        \label{eq:4}
        \min_{U} S(\rho_{P}(U)) = H(\{q_m\}) \equiv - \sum_m q_m \log_2 q_m.
    \end{equation}
    
    \textbf{Proof}: Let us take a basis of $\mathcal{H}_P$ as
    $\{|m\rangle, 1\le m\le d_P\}$, and a basis of $\mathcal{H}_R$ as
    $\{|\alpha\rangle, 1\le \alpha\le d_R\}$. Then we construct a unitary transformation
    \begin{equation}
        \label{eq:5}
        V |\psi_{(m-1)d_R + \alpha}\rangle = |m, \alpha\rangle
    \end{equation}
    Then
    \begin{align}
        \label{eq:6}
        \rho_P(V) & = \sum_{m,\alpha} p_{(m-1)d_R+\alpha} \Tr_R(|m,\alpha\rangle\langle m,\alpha|) \nonumber\\
        & = \sum_m q_m |m\rangle \langle m|
    \end{align}
    which implies that
    \begin{equation}
        \label{eq:7}
        S(\rho_P(V)) = H(\{q_{m}\})
    \end{equation}
    For any local unitary $U$, the eigen-decomposition of $\rho_P(U)$ is
    \begin{equation}
        \label{eq:8}
        \rho_P(U) = \sum_m q_m^{(U)} |m^{(U)}\rangle\langle m^{(U)}|
    \end{equation}
    Then
    \begin{align}
        \label{eq:9}
        \rho_P(U) & = \rho_P(WV) \nonumber\\
        & = \sum_{m,\alpha} p_{(m-1)d_R+\alpha} \Tr_R(W|m,\alpha\rangle\langle m,\alpha|W^{\dagger}) \nonumber\\
        & = \sum_{n,\beta} \sum_{m,\alpha}  p_{(m-1)d_R+\alpha} \langle n^{(U)},\beta|W|m,\alpha\rangle\langle m,\alpha|W^{\dagger}|n^{(U)},\beta\rangle |n^{(U)}\rangle\langle n^{(U)}| \nonumber\\
        & = \sum_{n,\beta} p^{(U)}_{n,\beta} |n^{(U)}\rangle \langle n^{(U)}|
    \end{align}
    where $W=UV^{-1}$ and
    \begin{align}
        \label{eq:10}
        p^{(U)}_{n,\beta} = \sum_{m,\alpha} B_{n,\beta;m,\alpha} p_{(m-1)d_R+\alpha}
    \end{align}
    with
    \begin{equation}
        \label{eq:11}
        B_{n,\beta;m,\alpha} = \langle n^{(U)},\beta|W|m,\alpha\rangle\langle m,\alpha|W^{\dagger}|n^{(U)},\beta\rangle
    \end{equation}
    Combining Eq.~\ref{eq:8} with Eq.~\ref{eq:9}, we arrives at
    \begin{equation}
        \label{eq:17}
        q_m^{(U)} = \sum_{\beta} p^{(U)}_{m,\beta}
    \end{equation}
    Notice that $B_{n,\beta;m,\alpha}\ge 0$ and
    \begin{equation}
        \label{eq:12}
        \sum_{n,\beta} B_{n,\beta;m,\alpha} = \sum_{m,\alpha} B_{n,\beta;m,\alpha} = 1
    \end{equation}
    Namely, $B$ is a doubly stochastic matrix. According to the Theorem~S1 above
    or Theorem II.1.9 in Bhatia's book~\cite{Bahtia1997}, we obtain
    \begin{equation}
        \label{eq:13}
        p^{(U)} \prec p
    \end{equation}
    Now we denote $p^{(U)}_{\downarrow}$ as $p^{(U)}$ in the decreasing order. Similarly, we
    introduce
    \begin{equation}
        \label{eq:14}
        q^{(U)}_{\downarrow,m} = \sum_{\alpha} p^{(U)}_{\downarrow,(m-1)d_R+\alpha}
    \end{equation}
    Then according to the definition of majorization,
    \begin{align}
        \label{eq:15}
        q^{(U)}_{\downarrow} & \prec q \\
        q^{(U)} & \prec q^{(U)}_{\downarrow}
    \end{align}
    which imples that $q^{(U)}\prec q$. Thus according to Theorem II.3.1 in Bhatia's
    book~\cite{Bahtia1997},
    \begin{equation}
        \label{eq:16}
        S(\rho_P(V)) \le S(\rho_P(U)).
    \end{equation}
    This completes our proof.
    
    In our context, $|\Phi\rangle$ is the ground state of the considered Hamiltonian $H$,
    $P$ is the target qubit that we aim to disentangle from the ground state,
    and $RP$ is the set of qubits that the unitary transformations $U$ for
    disentangling $P$ act on.
    
     \section{Reinforcement learning: Q-learning method}\label{Q-learning}
    
    The reinforcement learning process is a finite Markov decision
    process~\cite{sutton}. As shown in Fig.~\ref{fig:MDP}, this Markov process is
    described as a state-action-reward sequence, a state $S_t$ at time $t$ is
    transmitted into a new state $S_{t+1}$ together with giving a scalar reward
    $R_{t+1}$ at time $t+1$ by the action $A_t$ with the transmission probability
    $p(S_{t+1};R_{t+1}|S_t;A_t)$. For a finite Markov decision process, the sets of
    the states, the actions and the rewards are finite. The total discounted return
    at time $t$
    \begin{equation}
    	G_{t}=\sum_{k=0}^{\infty}\gamma^{k}R_{t+k+1},
    \end{equation}
    where $\gamma$ is the discount rate and $0\le\gamma\le1$.
    
    \begin{figure}[h]
    	\centering \includegraphics[width=0.85\columnwidth{},keepaspectratio]{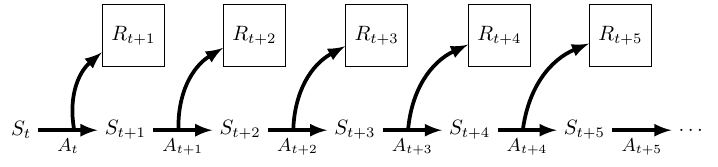}
    	\caption{\label{fig:MDP}A schematic diagram of Markov decision process.}
    \end{figure}
    
    The goal of an RL algorithm is to maximize the total discounted return for each
    state and action selected by the policy $\pi$, which is specified by a conditional
    probability of selecting an action $A$ for each state $S$, denoted as $\pi(A|S)$.
    
    In the Q-learning algorithm~\cite{Watkins1989}, the objective is to maximize the
    value of state-action function, the conditional discount return
    \begin{equation}
    	\label{eq:1}
    	Q_{\pi}(s,a) = E_{\pi}(G_{t}|S_{t}=s,A_{t}=a).
    \end{equation}
    
    \begin{figure}[h]
    	\centering \includegraphics[width=0.43\columnwidth{},keepaspectratio]{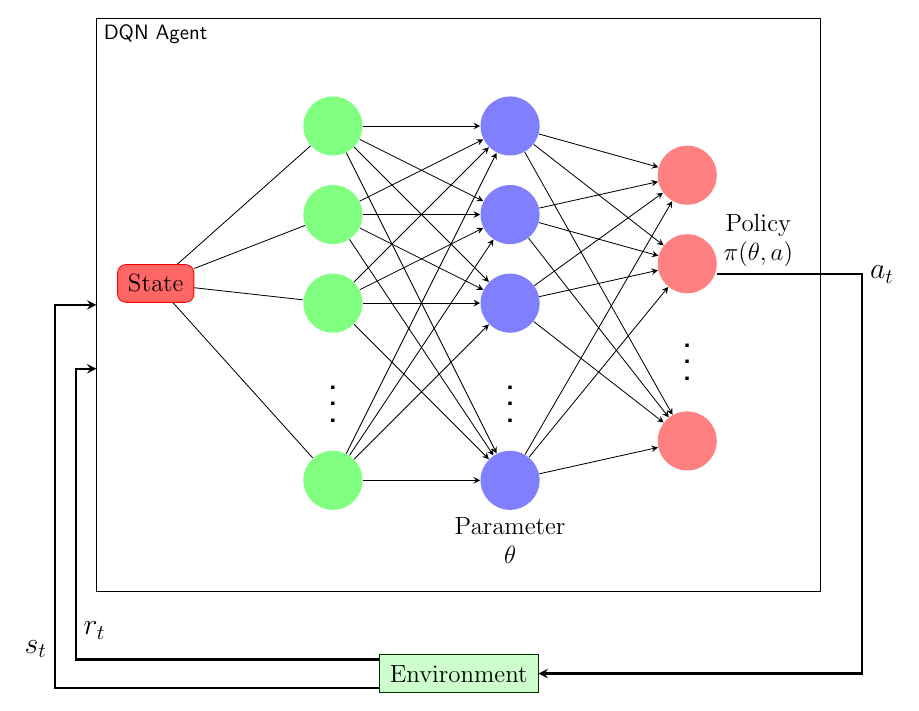}
    	\caption{\label{fig:DQN}A diagram of deep Q-network architecture.}
    \end{figure}
    
    If we have a policy $\pi$, then we calculate the value of state-action function
    $Q_{\pi}(S,A)$. For each state $s$, we take an action maximizing the value of
    state action $Q_{\pi}(s,A)$, which forms our new policy $\pi^{\prime}$. Then we calculate
    the value of state-action function $Q_{\pi^{\prime}}(S,A)$. Repeating the above
    procedure until the new policy equals the updated one, which is the optimal
    policy $\pi^{*}$ we are looking for.
    
    The update of Q-learning algorithm is defined as
    \begin{equation}
    	Q(S_{t},A_{t})  \leftarrow Q(S_{t},A_{t}) + \Delta Q
    	\label{Q}
    \end{equation}
    with
    \begin{equation}
    	\label{eq:5}
    	\Delta Q = \alpha[R_{t+1}+\gamma 
    	\max_{a}Q(S_{t+1},a)-Q(S_{t},A_{t})],
    \end{equation}
    where $\alpha$ is the step size parameter.
    
    As shown in Fig.~\ref{fig:DQN}, a deep Q-network is composed by a deep neural
    network, trained with a variant of Q-learning, whose input is state vector and
    whose output is a value function estimating future rewards. One significant
    improvement is that it designed a memory buffer to store the agent’s experiences
    and used them later. At each learning step, we fed the agent with a minibatch of
    experiences $\{S_{t},A_{t},R_{t+1},S_{t+1}\}$ with a modified prioritized
    experience replay (PER) method~\cite{PER}. The state $S_{t}$ is fed into the
    neural network to calculate the state-action value ${Q}(S_{t},A_{t};\theta)$. At the
    same time, a target Q-network is to calculate
    $\max\limits_{a'}{Q}(S_{t+1},a';\theta^{-})$ in Eq.~(\ref{loss}). At the end of each
    step of training, the evaluation network is updated through the back-propagation
    by minimizing the loss. Based on Eq.~(\ref{Q}), the loss is the mean square
    error (MSE) of the difference between the evaluation ${Q}(S_{t},A_{t};\theta)$ and
    the target $\max\limits_{a'}{Q}(S_{t+1},a';\theta^{-})$
    \begin{equation}
    	\rm{loss}=\rm{MSE}((R_{t+1} + 
    	\gamma\max\limits_{a'}{{Q}}(S_{t+1},a';\theta^{-}))-{Q}(S_{t},A_{t};\theta)).
    	\label{loss}	
    \end{equation} 
    
    During the learning episodes, the agent updates the
    parameters of the target network $\theta^{-} \rightarrow \theta$ to make better decisions.
    
     \section{The robustness of the training} 
     To confirm
    that the training scheme indeed extracts non-trivial structures 
    from the
    input state, we have checked the performance from the training 
    scheme when
    the quantum circuit is trained on different initial values and 
    states. To
    show the robustness of the RL-ansatz quantum circuit, we examined 
    our method
    with two example cases: (\romannumeral1) the initial parameters 
    of rotation
    gate extracted from the uniform distribution with mean value
    $\{0,\pi/2,\pi\}$ and (\romannumeral2) the disentangling circuit 
    training
    from different ground states of parameters ($\lambda$ or 
    $\Delta$).
    
    \begin{figure*}[h]
    	\centering \subfloat[]{%
    		\includegraphics[width=0.32\textwidth]{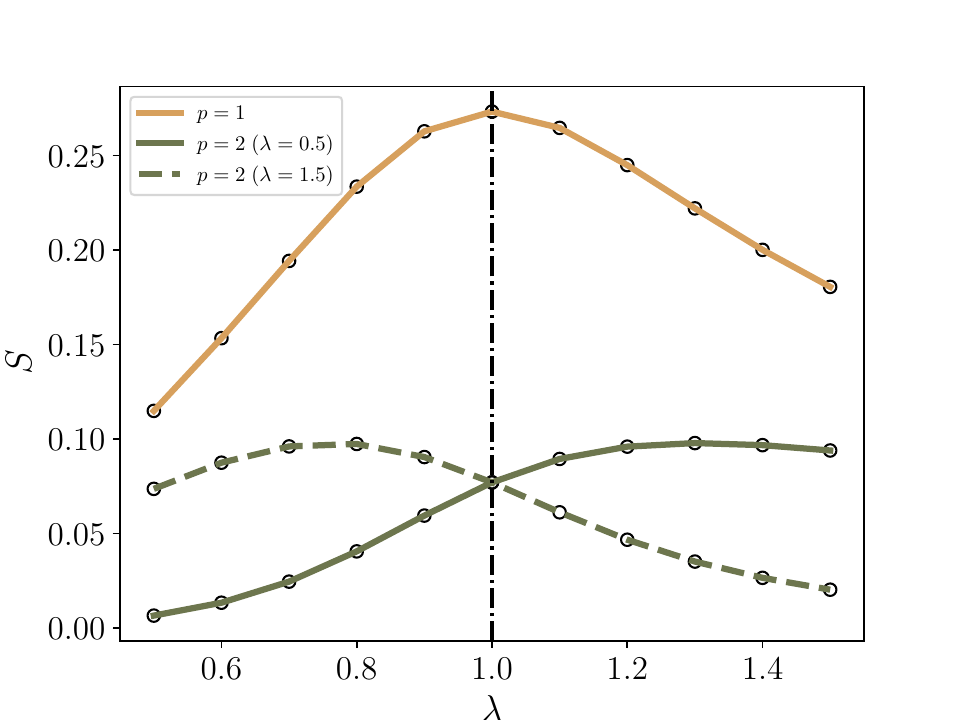}} 
    		\hspace{\fill}
    	\subfloat[]{%
    		\includegraphics[width=0.32\textwidth]{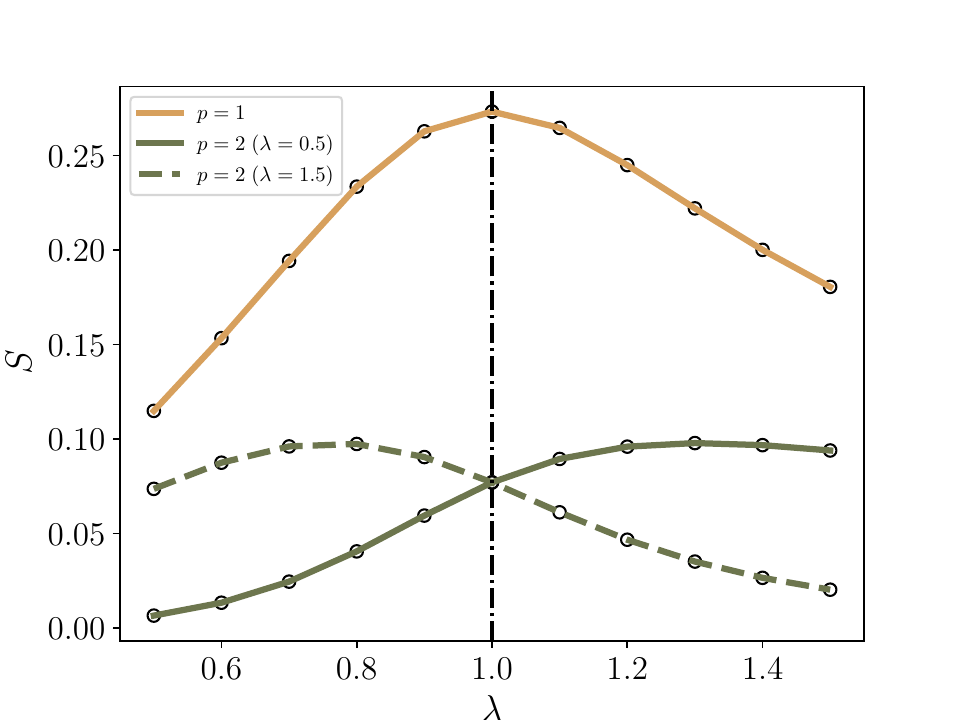}} 
    		\hspace{\fill}
    	\subfloat[]{%
    		\includegraphics[width=0.32\textwidth]{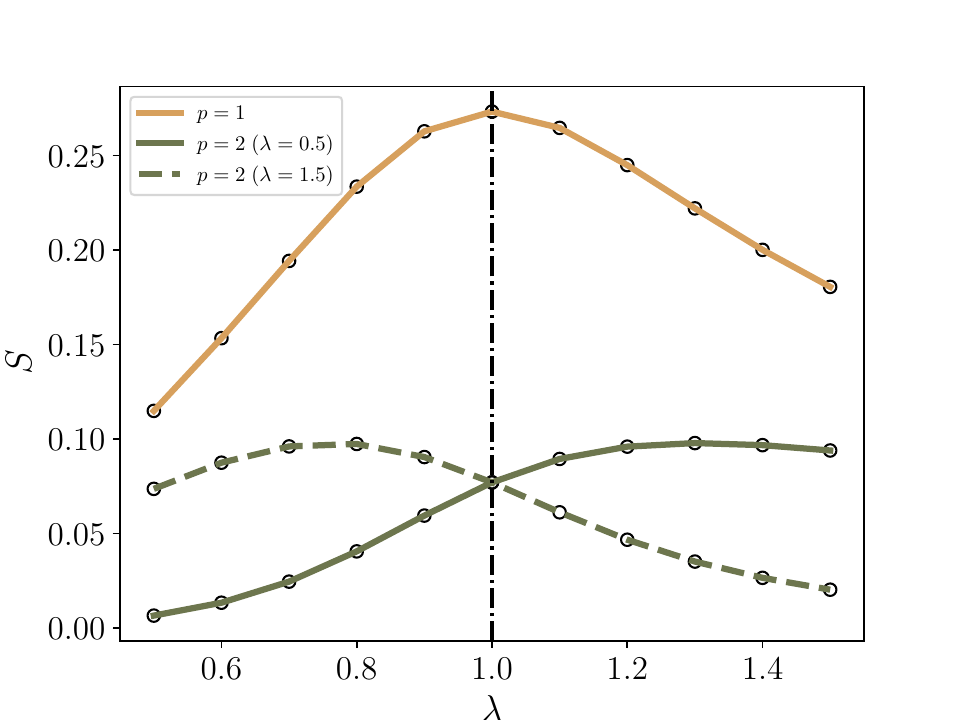}}\\
    	\caption{Entanglement entropy versus $\lambda$ with the 
    	RL-ansatz
    		disentangling circuit for TFIM with $N=8$ of 
    		(\romannumeral1) under
    		different initial values. (A): uniform distribution with 
    		mean value $0$
    		(B): uniform distribution with mean value $\pi/2$ (C): 
    		uniform
    		distribution with mean value $\pi$. The solid line 
    		denotes the
    		disentangling circuit training from $\lambda=0.5$; the 
    		dashed line
    		denotes the disentangling circuit training from 
    		$\lambda=1.5$.}
    	\label{fig:TFIMn}
    \end{figure*}
    
    \begin{figure*}[h]
    	\centering \subfloat[]{%
    		\includegraphics[width=0.32\textwidth]{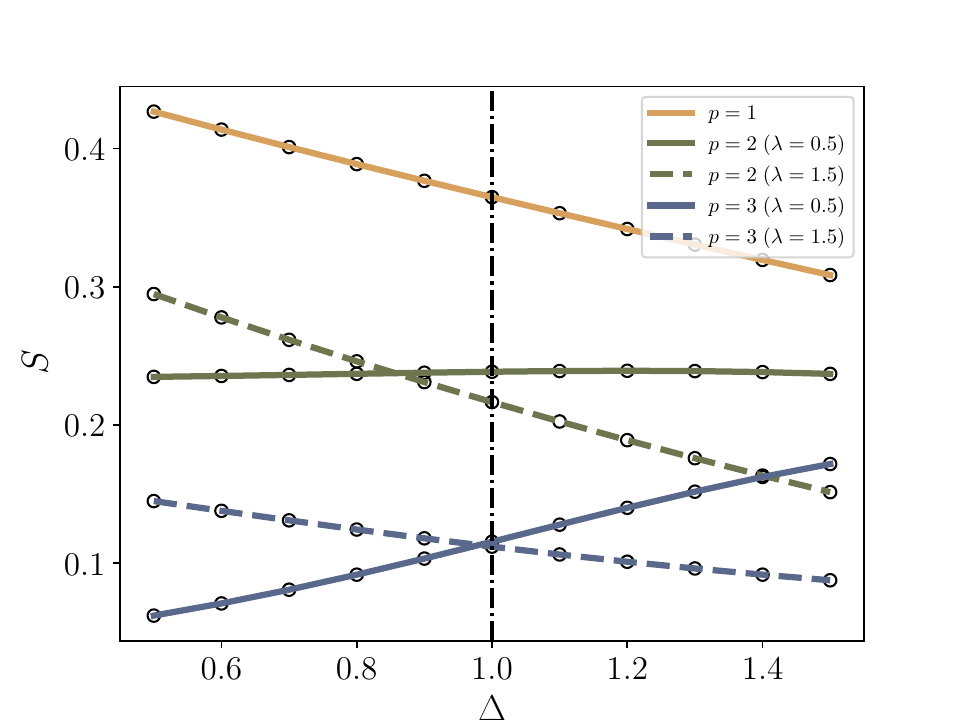}} 
    		\hspace{\fill}
    	\subfloat[]{%
    		\includegraphics[width=0.32\textwidth]{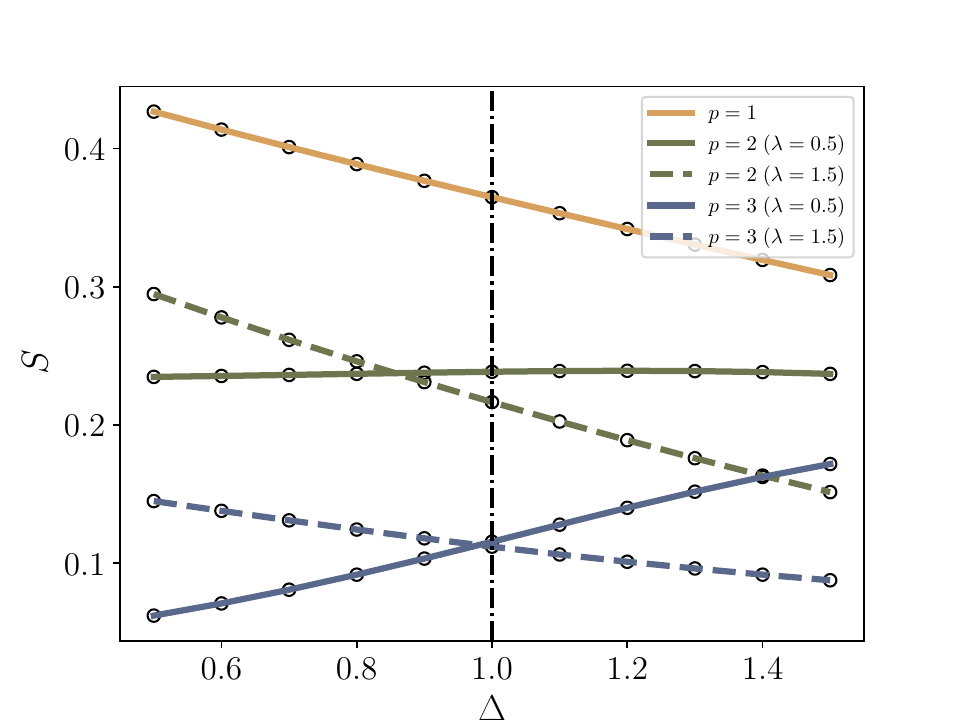}} 
    		\hspace{\fill}
    	\subfloat[]{%
    		\includegraphics[width=0.32\textwidth]{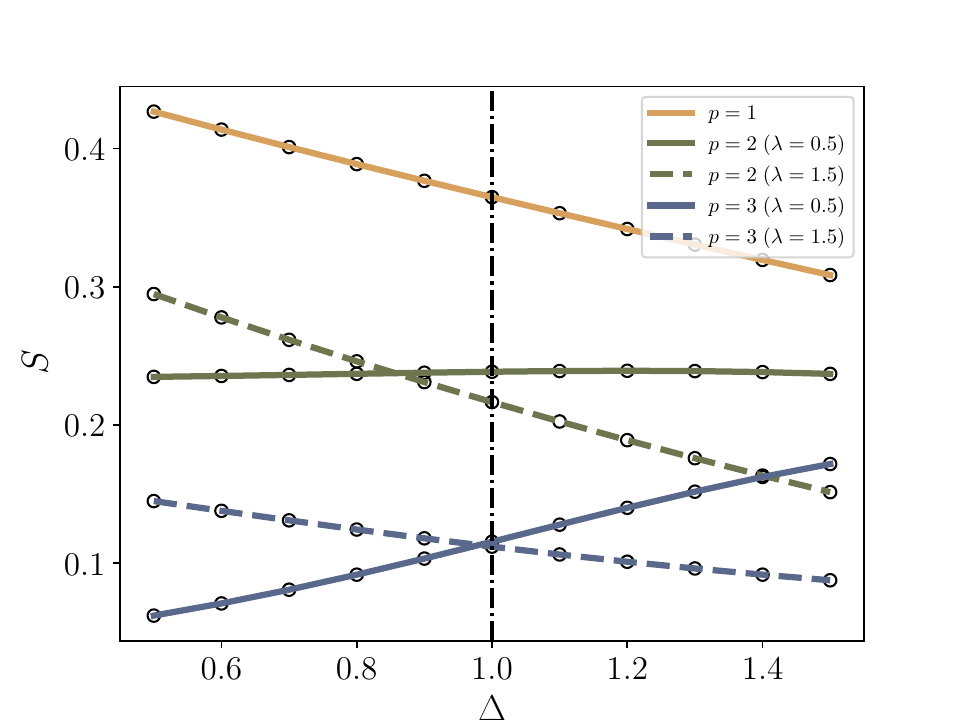}}\\
    	\caption{Entanglement entropy versus $\Delta$ with the 
    	RL-ansatz
    		disentangling circuit for XXZ model with $N=8$ of 
    		(\romannumeral1) under
    		different initial values. (A): uniform distribution with 
    		mean value $0$
    		(B): uniform distribution with mean value $\pi/2$ (C): 
    		uniform
    		distribution with mean value $\pi$. The solid line 
    		denotes the
    		disentangling circuit training from $\Delta=0.5$; the 
    		dashed line
    		denotes the disentangling circuit training from 
    		$\Delta=1.5$.}
    	\label{fig:XXZn}
    \end{figure*}

    \begin{figure*}[h]
    	\centering \subfloat[]{
    		\includegraphics[width=0.45\textwidth]{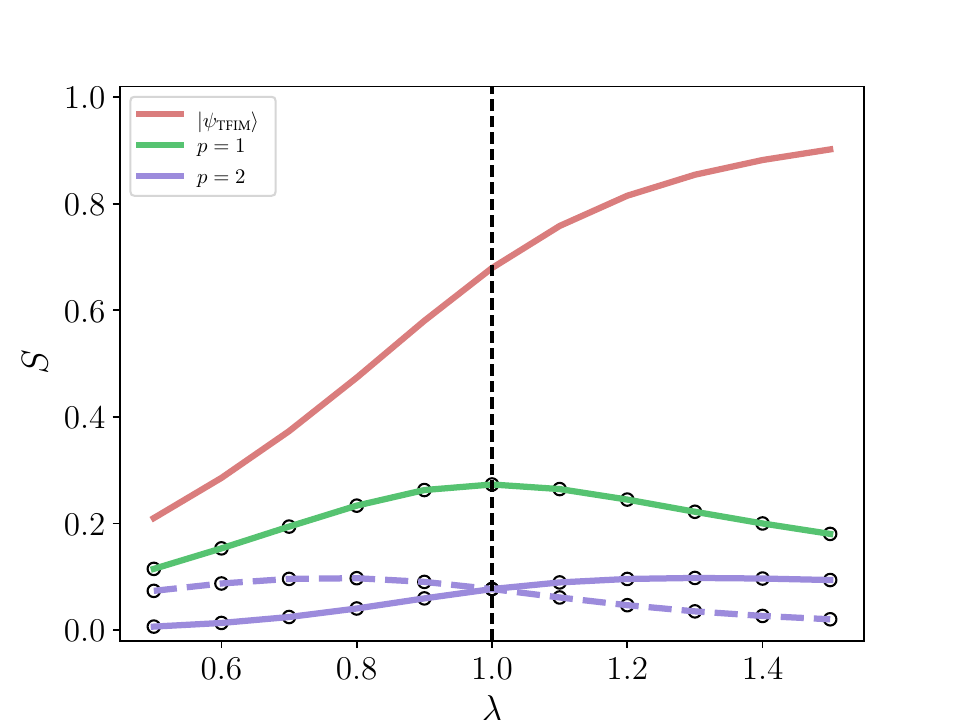} } 
    		\subfloat[]{
    		\includegraphics[width=0.45\textwidth]{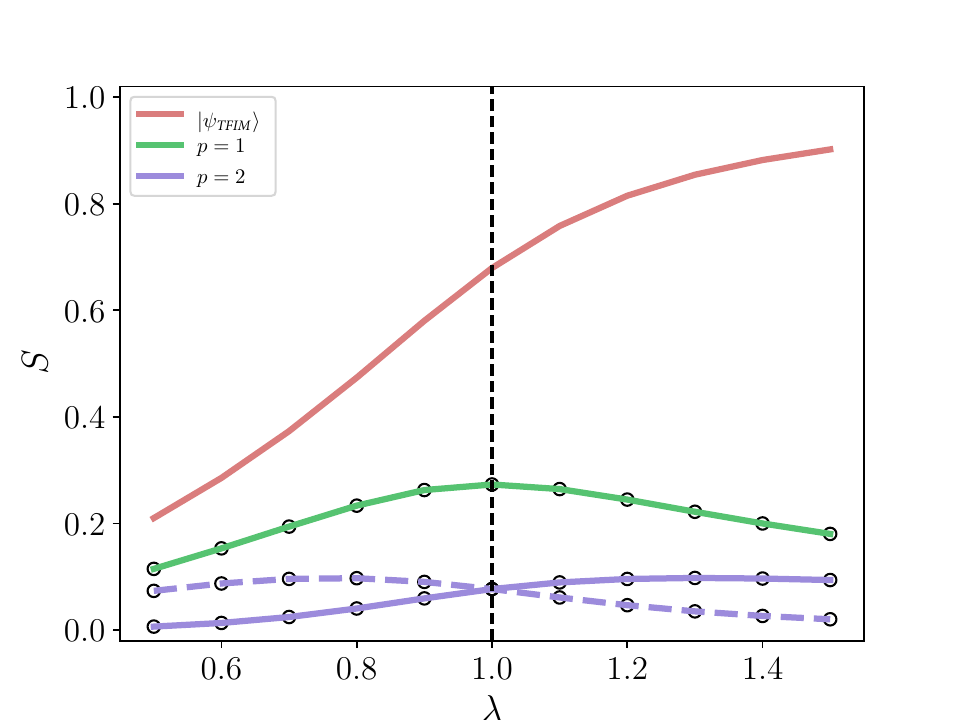} }
    	
    	\caption{Entanglement entropy versus $\lambda$ with the 
    	RL-ansatz
    		disentangling circuit for TFIM with $N=8$ of 
    		(\romannumeral2) under
    		different training states. (A): ground state with 
    		$\lambda=0.6, 1.3$
    		(B): ground state with $\lambda=0.7, 1.4$. The solid line 
    		denotes the
    		disentangling circuit training from $\lambda=0.6, 0.7$; 
    		the dashed line
    		denotes the disentangling circuit training from 
    		$\lambda=1.3, 1.4$.}
    	\label{fig:TFIMl}
    \end{figure*}
    
    \begin{figure*}[h]
    	\centering \subfloat[]{ 
    	\includegraphics[width=0.45\textwidth]{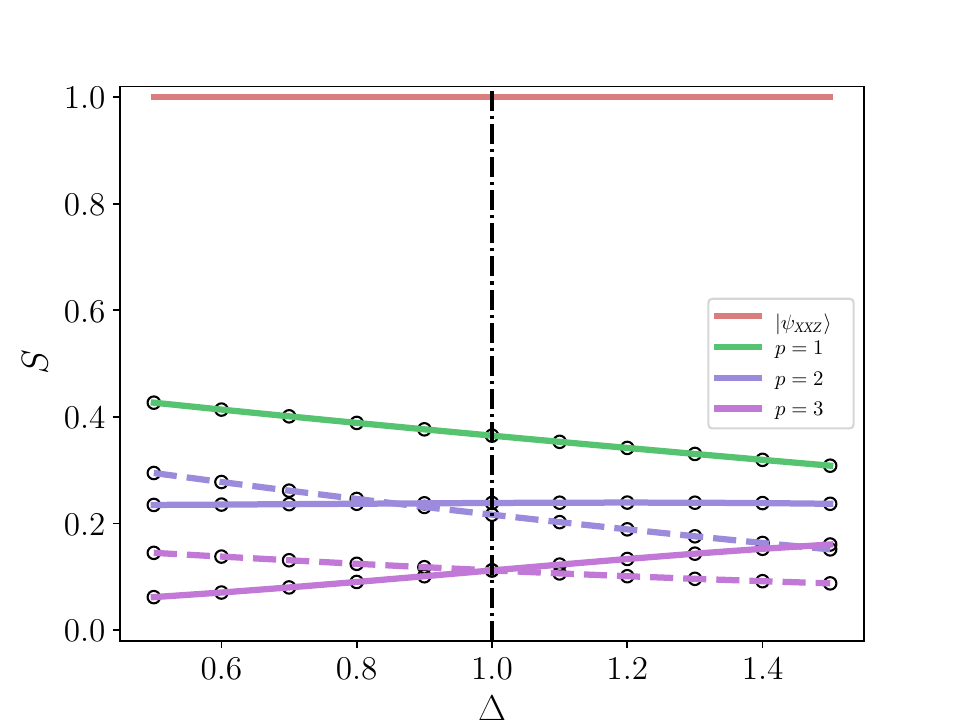}
    	} \subfloat[]{ 
    	\includegraphics[width=0.45\textwidth]{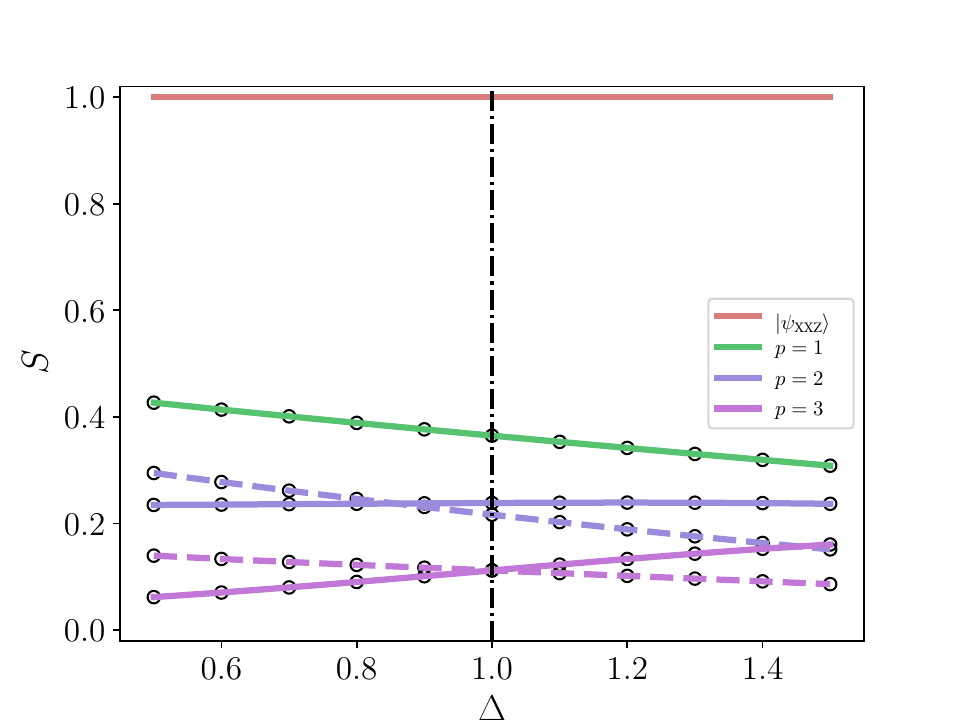} }
    	\caption{Entanglement entropy versus $\Delta$ with the 
    	RL-ansatz
    		disentangling circuit for XXZ model with $N=8$ of 
    		(\romannumeral2) under
    		different training states. (A): ground state with 
    		$\Delta=0.6, 1.3$ (B):
    		ground state with $\Delta=0.7, 1.4$. The solid line 
    		denotes the
    		disentangling circuit training from $\Delta=0.6, 0.7$; 
    		the dashed line
    		denotes the disentangling circuit training from 
    		$\Delta=1.3, 1.4$.}
    	\label{fig:XXZl}
    \end{figure*}
    
    As Fig.~\ref{fig:TFIMn} and Fig.~\ref{fig:TFIMl} show, the 
    algorithm's
    performance with the TFIM is robust under different conditions. 
    The same
    conclusion can be drawn in Fig.~\ref{fig:XXZn} and 
    Fig.~\ref{fig:XXZl} with
    the XXZ model. In other words, we may also look for the phase 
    transition at
    which the training is most independent of chosen conditions.
    
    \section{Concurrence in quantum phase transition}
In previous studies, a partial relationship between entanglement and quantum phase transitions has been found~\cite{Osterloh2002,PhysRevA.66.032110,gu2003entanglement,PhysRevLett.93.250404}.
Here we show the relationship between concurrence and quantum phase transition, which is used to compare with our method. The concurrence~\cite{wootters1998entanglement} between sites $i$ and $j$ is defined as
\begin{equation}
    C(i, j)=\max \left\{r_{1}(i, j)-r_{2}(i, j)-r_{3}(i, j)-r_{4}(i, j), 0\right\},
\end{equation}
where $r_\alpha(i,j)$ are the square roots of the eigenvalues of the
product matrix $\rho(i, j) \tilde{\rho}(i, j)$ in descending order; the spin
flipped matrix is defined as $\tilde{\rho}=\sigma^{y} \bigotimes \sigma^{y} \rho^{*} \sigma^{y} \bigotimes \sigma^{y}$.

Based on previous studies, we know that the next nearest neighbor concurrence $C(2)$ is related to the quantum phase transition of the TFIM model. The relationship is shown in Fig.~\ref{fig:curr} (a). In the XXZ model, the nearest-neighbor concurrence $C(1)$ is related to its quantum phase transition, and the relationship is shown in Fig.~\ref{fig:curr} (b). From Fig.~\ref{fig:curr} (a), it can be seen that as the spin number $N$ of the system increases, the mutation region of concurrence becomes narrower near the critical point. However, this method can only give a phase transition range in a finite size of systems and cannot precisely determine the critical point.  
As can be seen in Fig.~\ref{fig:curr} (b), the maximum value of concurrence is at the critical point of the XXZ model. However, unlike TFIM, it is independent of the system size. We can conclude that there is no general method to observe only the change in the quantity of entanglement to determine the quantum phase transition.
    \begin{figure*}
    \hspace{0.66cm}
    \begin{subfigure}[h]{0.407\textwidth}
        \includegraphics[width=\linewidth]{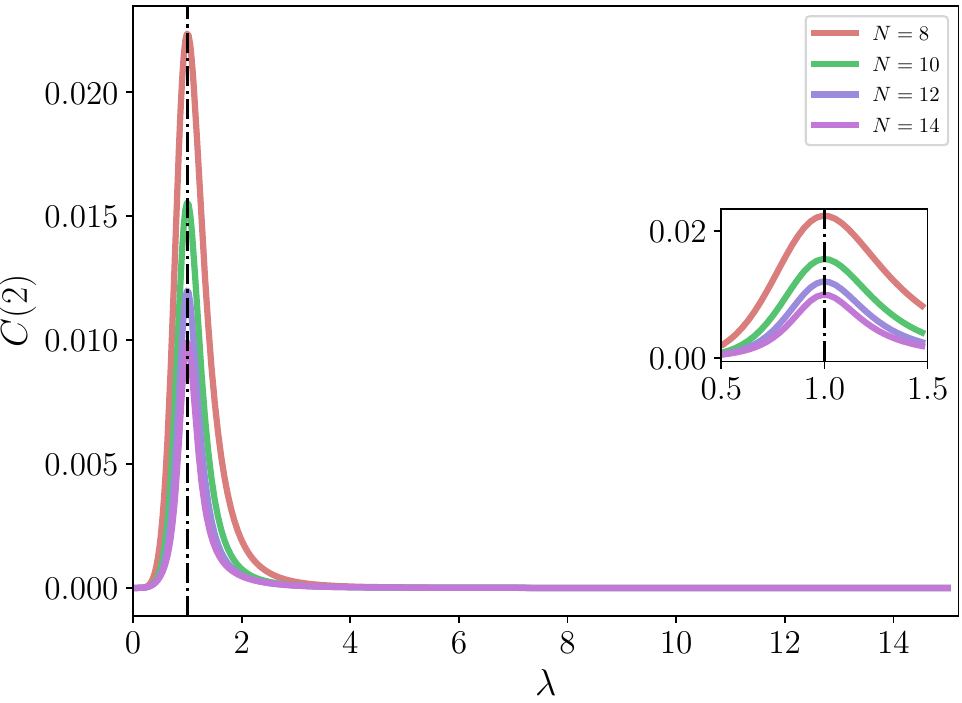}
        \caption{}
    \end{subfigure}
    \hspace{0.84cm}
    \begin{subfigure}[h]{0.4\textwidth}
        \includegraphics[width=\linewidth]{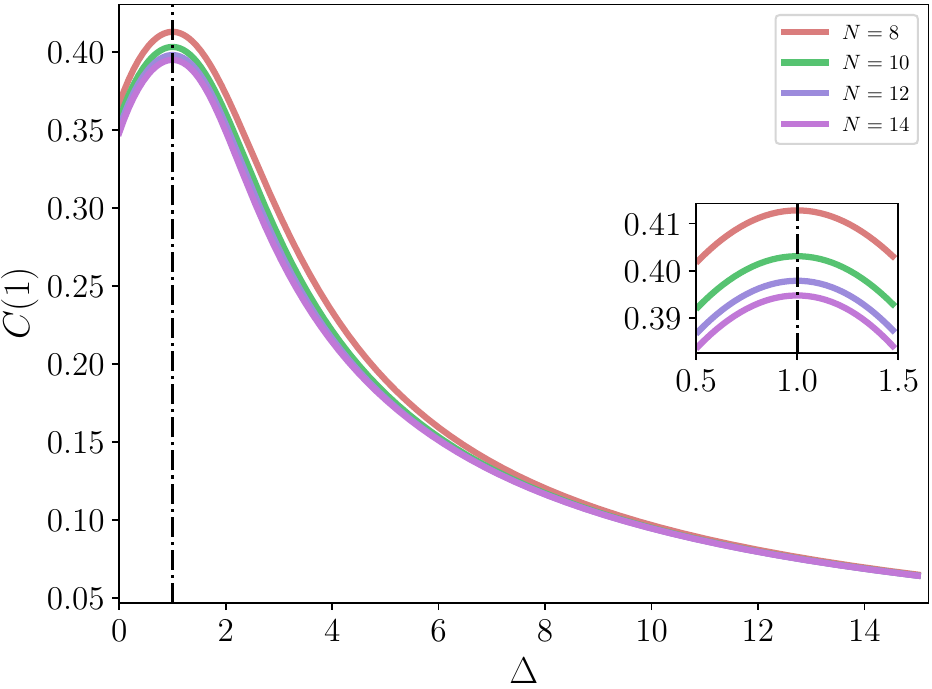}
        \caption{}
    \end{subfigure}
    
    \caption{Concurrence as function of parameters. (a): next nearest neighbor concurrence $C(2)$ in TFIM with $N$ sites (b): nearest-neighbor concurrence $C(1)$ in XXZ model with $N$ sites.}
    \label{fig:curr}
\end{figure*}
 
\section{DMRG in phase identification} 

Fig.~\ref{DMRG} presented illustrates the results obtained from the Density Matrix Renormalization Group (DMRG) method in determining phase transitions in the Transverse Field Ising Model (TFIM) and the XXZ model. We evaluated the performance of two critical physical quantities in the DMRG ground state. The first quantity is the magnetization, which corresponds to the average of the pauli operator observations in the z-direction over all sites. The second quantity is the entanglement entropy, calculated by considering the partition between half of the chain and the remaining half.

	\begin{figure*}[h]
 \hspace{0.65cm}
	\subfloat[]{
		\includegraphics[width=0.404\textwidth]{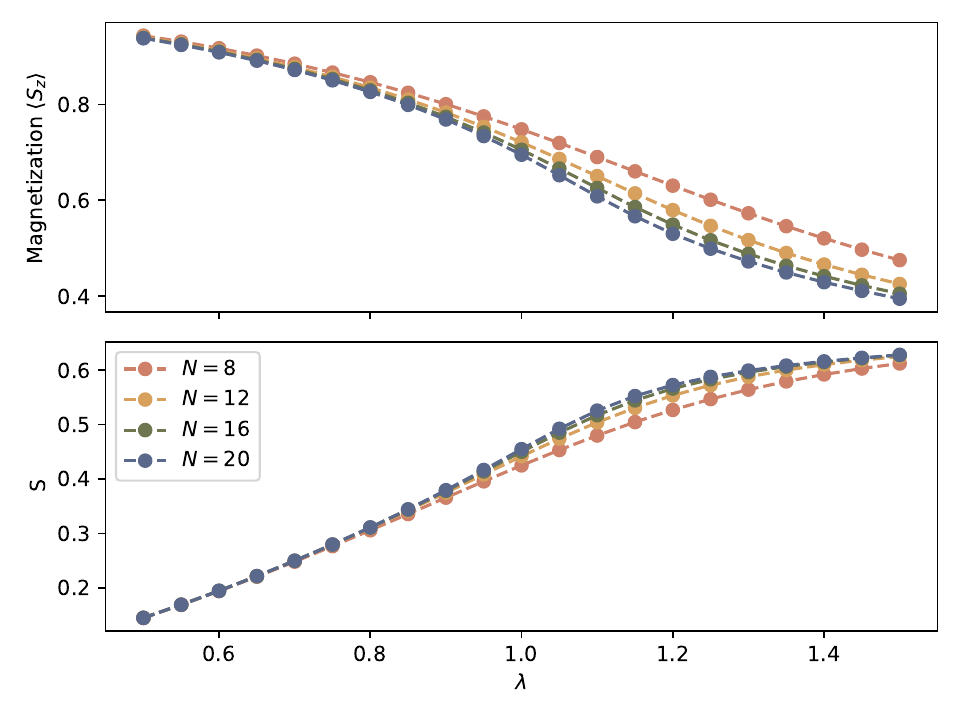} } 
    \hspace{0.45cm}
	\subfloat[]{
		\includegraphics[width=0.42\textwidth]{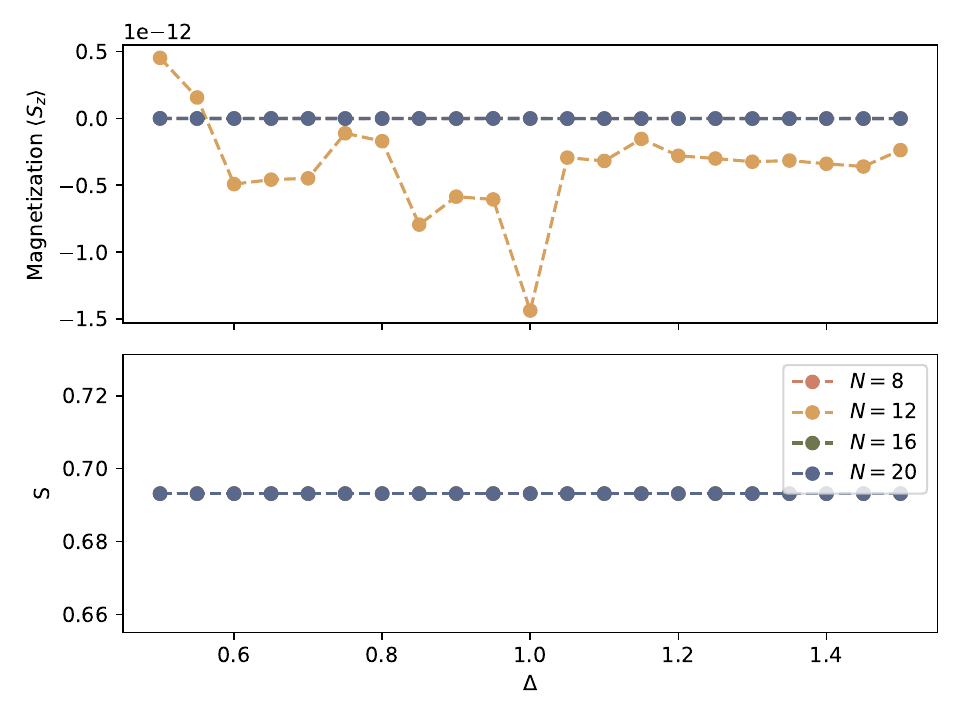} }
	\caption{Ground state magnetization and entanglement entropy for (a) Transverse Field Ising Model (TFIM) and (b) XXZ model with varying system sizes. The top subplots display the magnetization $\langle S_z \rangle$ as a function of the system parameter, while the bottom subplots show the entanglement entropy $S$ as a function of the system parameter.}\label{DMRG}
\end{figure*}

From the analysis of the figure, it becomes apparent that pinpointing the precise phase transition point using the DMRG method is a challenging task. Nevertheless, the primary advantage of employing DMRG lies in its ability to calculate the ground state of larger systems with remarkable accuracy. 

\section{Local Order Parameters and Phase Transition Indicators}

    For the finite 1D transverse-field Ising model (TFIM) presented in our example, the limited system size introduces significant boundary effects. As illustrated in Fig.~\ref{order}(a), magnetization changes smoothly as a function of the transverse field $\lambda$, contrasting with the sharp transitions observed in infinite systems. This smoothness arises from the lack of a true phase transition in finite systems.

Turning to the finite 1D XXZ spin chain, we agree that the antiferromagnetic (AFM) correlation function
$\left\langle S_i^z S_{i+1}^z\right\rangle$ conveys important information about the system's phase. It's crucial to recognize that the transition at $\Delta=1$, as discussed in our manuscript, is of the Kosterlitz-Thouless (KT) type. This infinite order phase transition does not feature a discontinuous change in any local order parameter, which can render the transition's signatures less pronounced in finite-size systems. While Fig.~\ref{order}(b) demonstrates that the AFM correlation function for pinpointing the transition at $\Delta = 1$ remains a challenge.

\begin{figure}[h]
	\centering \subfloat[]{
		\includegraphics[width=0.45\textwidth]{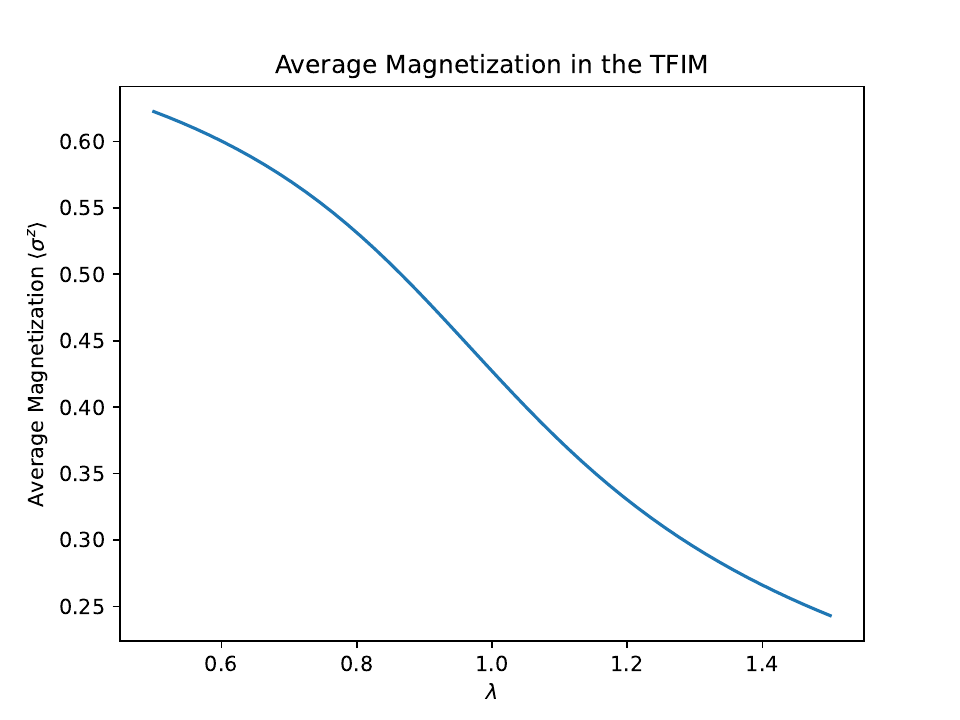} } 
	\subfloat[]{
		\includegraphics[width=0.45\textwidth]{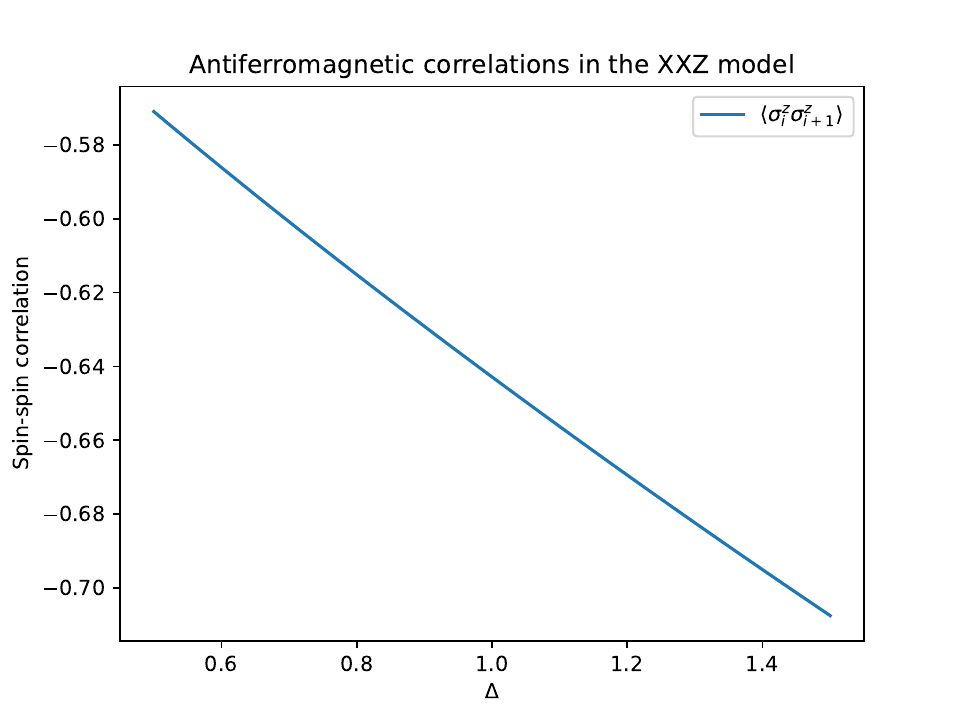} }
	\caption{compare the reduced ground states with 2 qubits via magnetization or correlation function with the models of size $N=8$ in our manuscript.}\label{order}
\end{figure}

The findings discussed above align with the Density Matrix Renormalization Group (DMRG) results we presented in our previous analysis (refer to Fig.~\ref{DMRG}). Although the order parameter in a finite system provides valuable insights into the system's state and its proximity to phase transitions, it alone cannot precisely pinpoint the critical point due to the inherent limitations posed by finite size effects. A more accurate determination of the critical point can be achieved through the application of scaling laws and by examining the behavior of the system across various sizes.

Furthermore, our approach serves as an exemplary case of how our methodology can be differentiated from traditional methods used in identifying quantum phase transitions. This demonstrates that our technique not only complements existing phase transition identification methods but also enhances the robustness and depth of analysis in studying critical phenomena in finite systems.

\section{Ground state overlap measures}
Additionally, we evaluated the ground state overlap for the two models under investigation. As depicted in Fig.~\ref{GSO}, similar to the expectation or correlation methods, this approach does not yield precise information about the phase transition point in the models we analyzed. Concerning alternative methods that employ entanglement to identify phase transitions, such as the concurrence method outlined in the supplementary material, we suggest that relying exclusively on changes in entanglement measures is not a universally dependable indicator of quantum phase transitions. This observation underscores the necessity for a multifaceted approach when analyzing and identifying critical points in quantum systems.

\begin{figure}[h]
	\centering \subfloat[]{
		\includegraphics[width=0.45\textwidth]{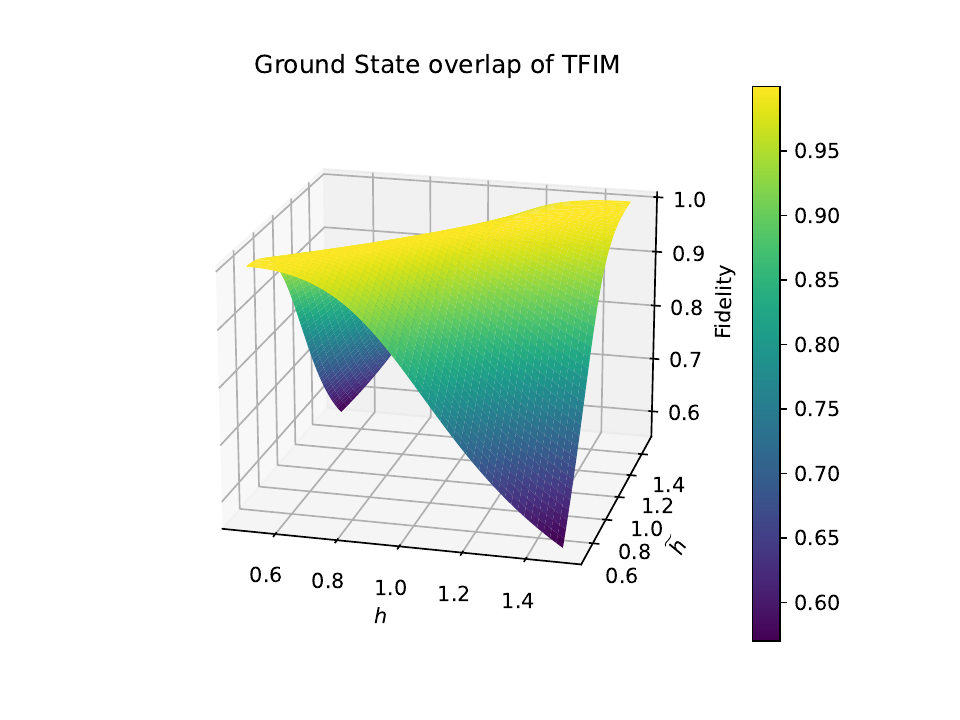} } 
	\subfloat[]{
		\includegraphics[width=0.45\textwidth]{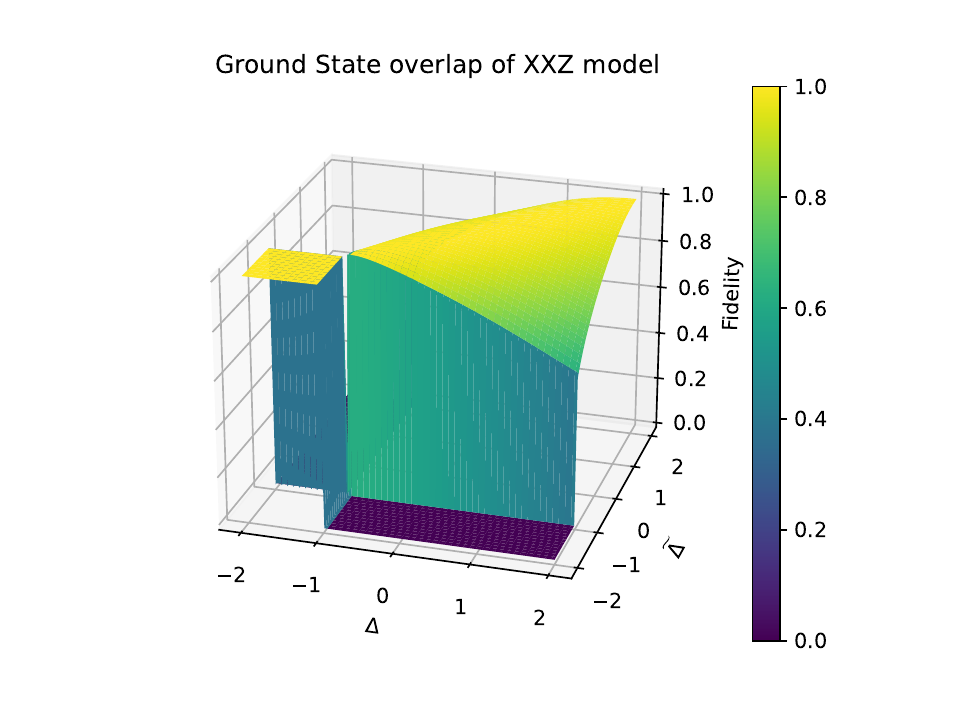} }
	\caption{Compare the ground states via Ground state overlap with the models in our manuscript.}
\end{figure}\label{GSO}

 \section{The convergency of RL-ansatz disentangling circuit }\label{conv}
    
    As Fig.~\ref{fig:TFIMc} and Fig.~\ref{fig:XXZc} shown, the 
    performance of the
    RL-ansatz disentangling circuit is convergence with the increase of circuit layers.
    
    \begin{figure*}[tbhp]
    \centering
            \begin{subfigure}[h]{0.39\textwidth}
            \includegraphics[width=\linewidth]{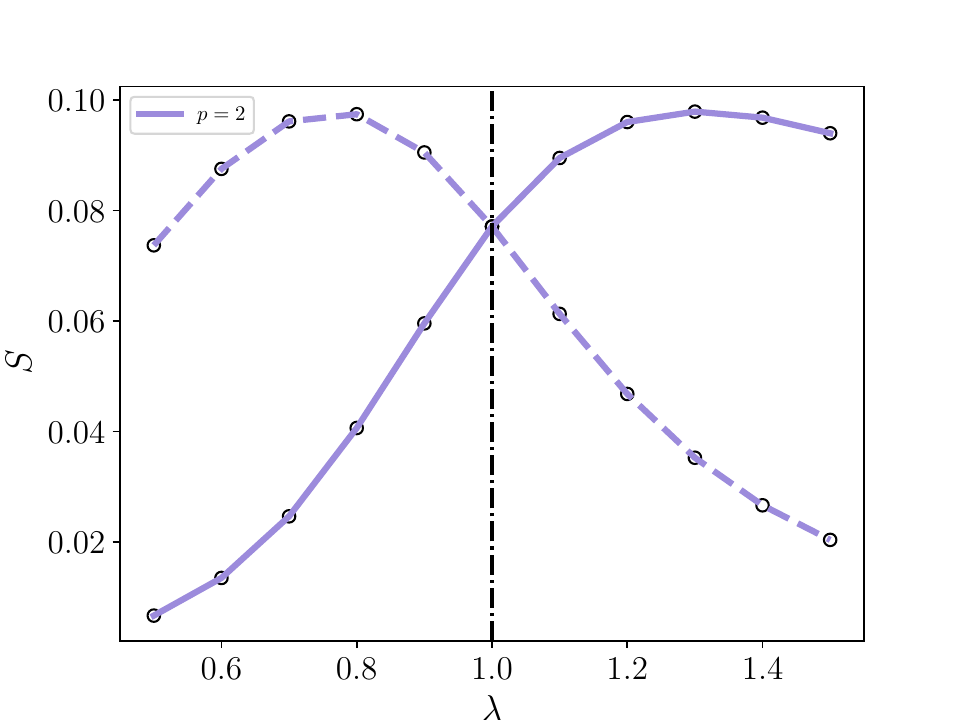}
            \caption{}
        \end{subfigure}
        \hspace{1cm}
        \begin{subfigure}[h]{0.39\textwidth}
            \includegraphics[width=\linewidth]{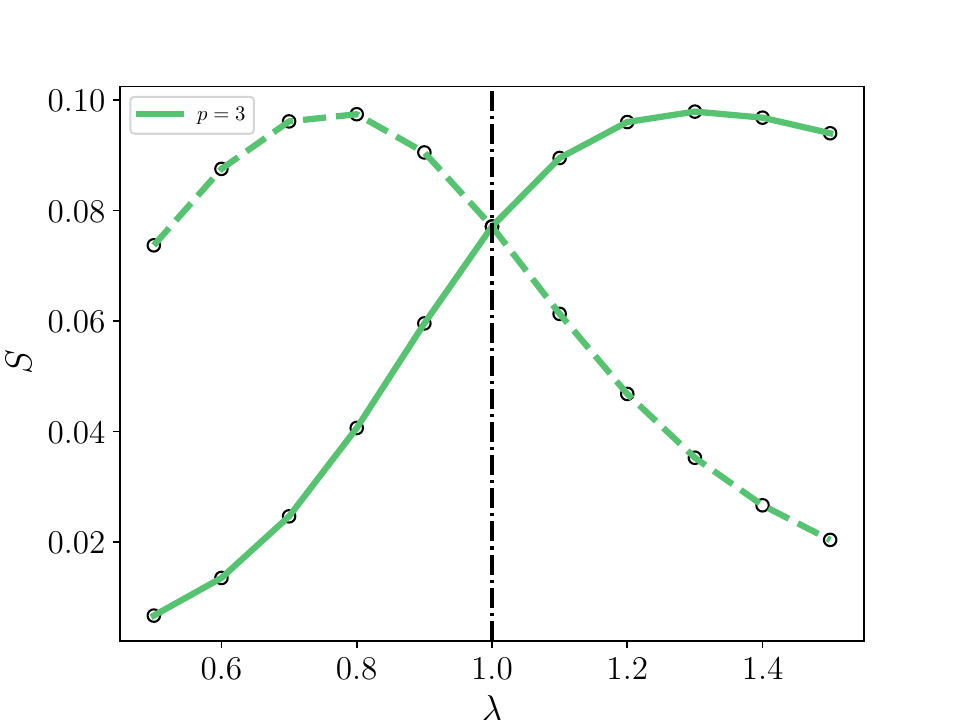}
            \caption{}
        \end{subfigure}
        \caption{Entanglement entropy versus $\lambda$ for 8-site TFIM.
            $P$ is the number of circuit layers. (a): circuit layers of $p=2$ (b): circuit layers
            of $p=3$. The solid line denotes the disentangling circuit training from $\lambda=0.5$;
            the dashed line denotes the disentangling circuit training from $\lambda=1.5$.}
        \label{fig:TFIMc}
    \end{figure*}

    \begin{figure*}[htbp]
    \centering
            \begin{subfigure}[h]{0.39\linewidth}
    \includegraphics[width=\linewidth]{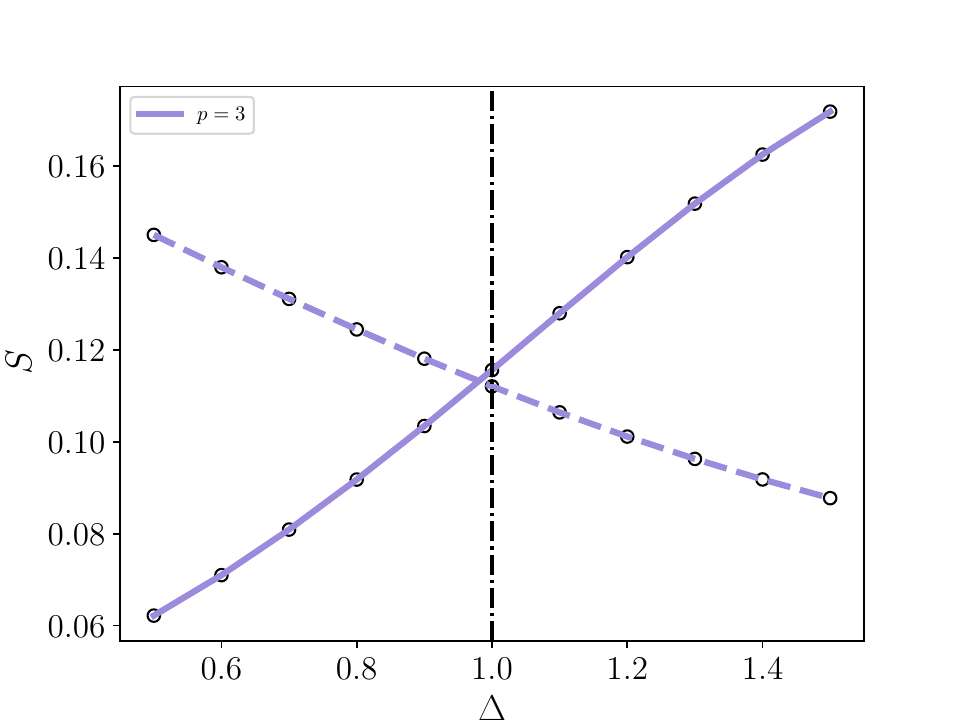}
    \caption{}
\end{subfigure}
\hspace{1cm}
\begin{subfigure}[tbh]{0.39\linewidth}
    \includegraphics[width=\linewidth]{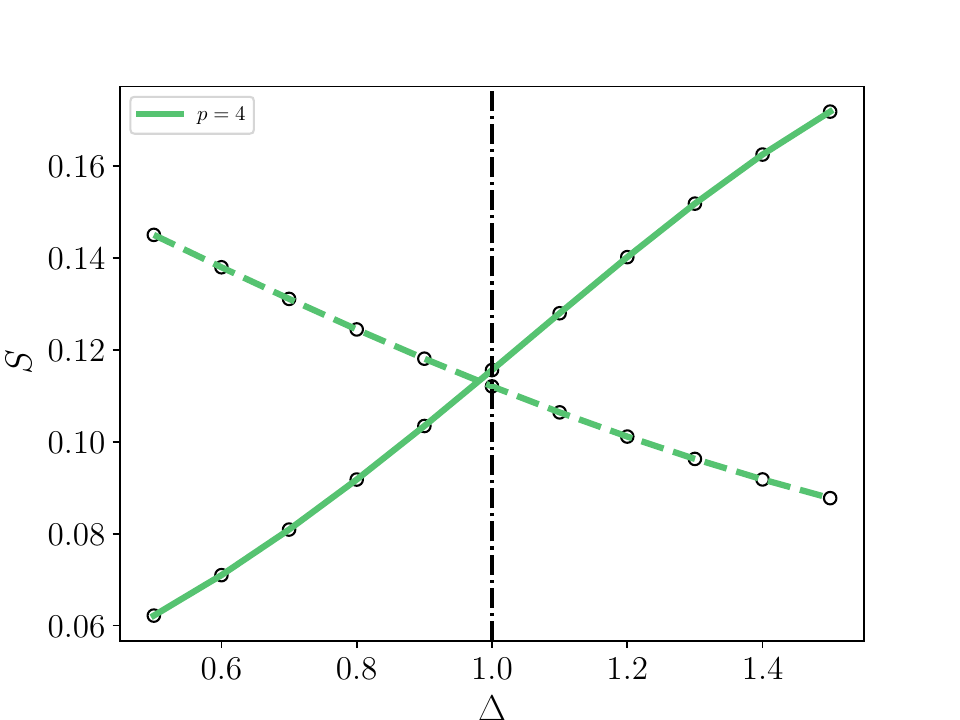}
    \caption{}
\end{subfigure}

        \caption{Entanglement entropy versus $\Delta$ for 8-site XXZ model. $P$ is
            the number of circuit layers (a): circuit layers of $p=3$ (b): circuit layers
            of $p=4$. The solid line denotes the disentangling circuit training from $\Delta=0.5$;
            the dashed line denotes the disentangling circuit training from  $\Delta=1.5$.}
        \label{fig:XXZc}
    \end{figure*}

\begin{figure}[h]
	\centering \subfloat[]{
		\includegraphics[width=0.39\textwidth]{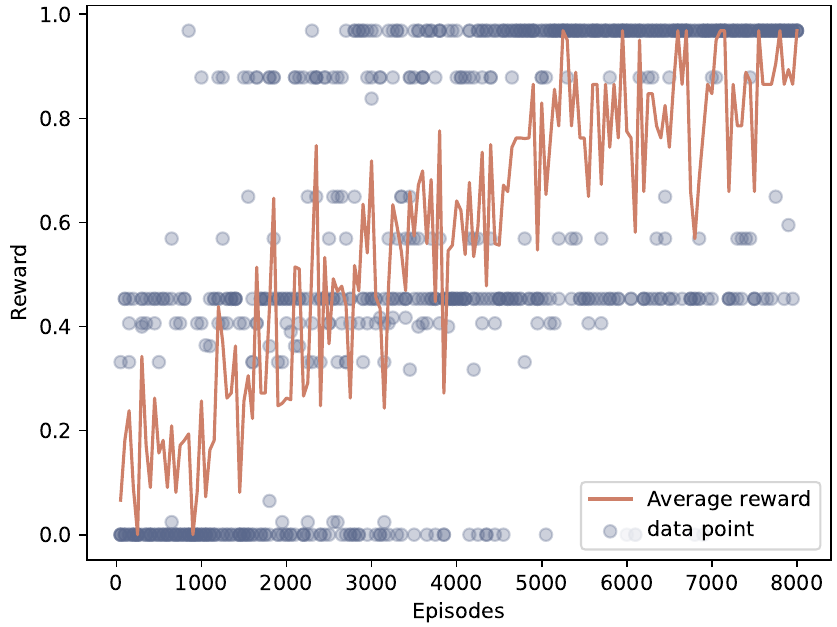} } 
  \hspace{1cm}
	\subfloat[]{
		\includegraphics[width=0.39\textwidth]{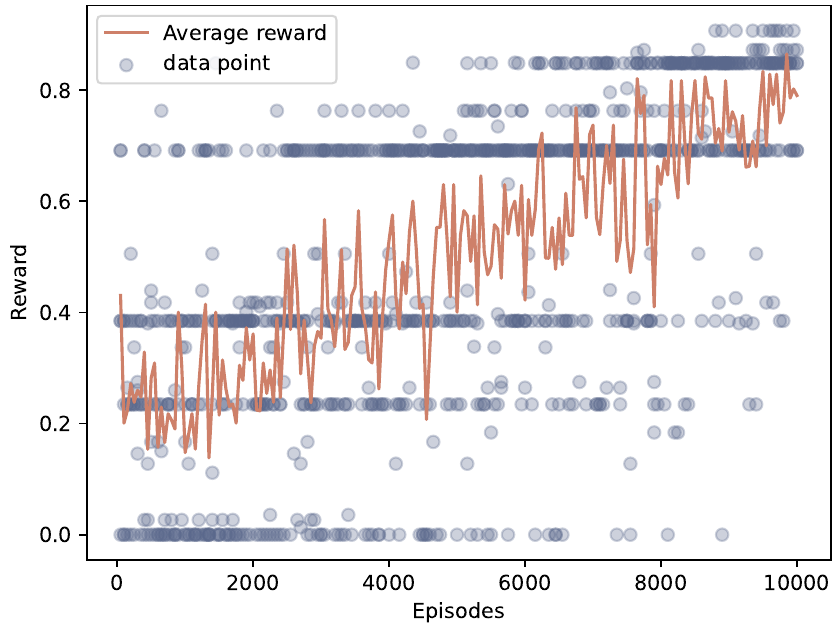} }
	\caption{Rewards for a reinforcement learning (RL) agent over a series of episodes in five different runs, applied to (a) Transverse Field Ising Model (TFIM) with 8 sites and a quantum circuit layer $p=2$, and (b) XXZ model with 8 sites and a quantum circuit layer $p=3$. Individual data points are depicted as semi-transparent blue dots, while the average reward curve is represented by a solid red line.\label{rewards}}
\end{figure}

Fig.~\ref{rewards} presents the rewards monitored over episodes for 5 agents applied to the TFIM and XXZ models. We observe that the average reward across the 5 agents increases, indicating that our algorithm successfully learns to design disentanglement quantum circuits. However, it is crucial to note that the convergence of a DQN is not solely determined by the rewards. Due to the off-line learning and epsilon-greedy exploration properties in DQN, rewards are influenced by both the current policy and the experiences stored in the memory buffer. During the learning process, a random action selection weight typically starts with a high value and gradually decays over time. This enables the agent to explore the environment more during the initial stages and rely more on its learned Q-values as it gains experience. Consequently, the rewards might not exhibit smooth behavior.

    \section{Hyper-parameters in RL algorithm}
    
    Our RL agent makes use of a deep neural network to approximate the Q values for
    the possible actions of each state. The network consists of 6 layers of the
    neural network. All layers have ReLU activation functions except the output
    layer which has linear activation. The hyper-parameters of the network are
    summarized in Table~\ref{para}.
    \begin{table}[H]
        \centering
        \caption{Training Hyper-Parameters}
        \begin{tabular}{cc}
            \textrm{Hyper-parameter} &\textrm{Values} \\
            \hline
            Neurons in neural network & $\{512,512,512,512,512,512\}$\\

            Minibatch size & 120 \\
            
            Replay memory size & 100000\\
            
            Learning rate & $0.0001$\\
            
            Update period & 100\\
            
            Reward decay $\gamma$ & 0.99\\

            \label{para}
            
        \end{tabular}
    \end{table}

\end{document}